\documentclass[draft]{article} 
\usepackage{aaai2026}  
\usepackage{times}  
\usepackage{helvet}  
\usepackage{courier}  
\usepackage[hyphens]{url}  
\usepackage{graphicx} 
\usepackage{natbib}  
\usepackage{caption} 
\usepackage{algorithm}
\usepackage[noend]{algorithmic}

\usepackage{newfloat}
\usepackage{listings}
\DeclareCaptionStyle{ruled}{labelfont=normalfont,labelsep=colon,strut=off} 
\floatstyle{ruled}
\newfloat{listing}{tb}{lst}{}
\floatname{listing}{Listing}
\usepackage{amsmath}

\title{Improved Differentially Private Algorithms for Rank Aggregation}
\author{
    Written by AAAI Press Staff\textsuperscript{\rm 1}\thanks{With help from the AAAI Publications Committee.}\\
    AAAI Style Contributions by Pater Patel Schneider,
    Sunil Issar,\\
    J. Scott Penberthy,
    George Ferguson,
    Hans Guesgen,
    Francisco Cruz\equalcontrib,
    Marc Pujol-Gonzalez\equalcontrib
}

\author {
    Quentin Hillebrand\equalcontrib\textsuperscript{\rm 1},
    Pasin Manurangsi\equalcontrib\textsuperscript{\rm 2},
    Vorapong Suppakitpaisarn\textsuperscript{\rm 3},
    Phanu Vajanopath\equalcontrib\textsuperscript{\rm 4}
}
\affiliations {
    \textsuperscript{\rm 1}University of Copenhagen, Denmark\\
    \textsuperscript{\rm 2}Google Research, Thailand\\
    \textsuperscript{\rm 3}The University of Tokyo, Japan\\
    \textsuperscript{\rm 4}University of Wrocław, Poland
}

\usepackage{amsmath}
\usepackage{amsthm}
\usepackage{amssymb}
\usepackage{dsfont}
\usepackage{cleveref}

\usepackage{xcolor}
\definecolor{Gred}{RGB}{219, 50, 54}
\definecolor{Ggreen}{RGB}{60, 186, 84}
\definecolor{Gblue}{RGB}{72, 133, 237}
\definecolor{Gyellow}{RGB}{247, 178, 16}
\definecolor{ToCgreen}{RGB}{0, 128, 0}
\definecolor{myGold}{RGB}{231,141,20}
\definecolor{myBlue}{rgb}{0.19,0.41,.65}
\definecolor{myPurple}{RGB}{175,0,124}

\newtheorem{theorem}{Theorem}[section]
\newtheorem{corollary}[theorem]{Corollary}

\newtheorem{lemma}[theorem]{Lemma}

\theoremstyle{definition}
\newtheorem{definition}[theorem]{Definition}

\newtheorem{observation}[theorem]{Observation}

\newcommand{\tgamma}{\tilde{\gamma}}

\newcommand{\N}{\mathbb{N}}

\newcommand{\bone}{\mathbf{1}}
\newcommand{\R}{\mathbb{R}}
\newcommand{\bw}{\mathbf{w}}

\newcommand{\bx}{\mathbf{x}}

\newcommand{\bW}{\mathbf{W}}

\newcommand{\bbS}{\mathbf{S}}
\newcommand{\cX}{\mathcal{X}}

\newcommand{\cB}{\mathcal{B}}

\newcommand{\eps}{\varepsilon}

\newcommand{\Lap}{\mathrm{Lap}}
\newcommand{\cN}{\mathcal{N}}
\newcommand{\agg}{\mathrm{agg}}
\newcommand{\tv}{\tilde{v}}
\newcommand{\tu}{\tilde{u}}

\newcommand{\cost}{\mathrm{cost}}
\newcommand{\cZ}{\mathcal{Z}}
\newcommand{\cE}{\mathcal{E}}
\newcommand{\Zl}{\cZ_{I}}
\newcommand{\Zs}{\cZ_{B}}
\newcommand{\tbw}{\tilde{\bw}}
\newcommand{\tw}{\tilde{w}}
\newcommand{\out}{\mathrm{out}}

\newcommand{\bs}{\mathbf{s}}
\newcommand{\bt}{\mathbf{t}}
\newcommand{\tbs}{\tilde{\bs}}
\newcommand{\tbt}{\tilde{\bt}}
\newcommand{\ttt}{\tilde{t}}

\newcommand{\bX}{\mathbf{X}}

\newcommand{\E}{\mathbb{E}}
\newcommand{\cP}{\mathcal{P}}
\newcommand{\tq}{\tilde{q}}
\newcommand{\cA}{\mathcal{A}}
\newcommand{\cY}{\mathcal{Y}}

\newcommand{\cO}{\mathcal{O}}
\newcommand{\cR}{\mathcal{R}}
\DeclareMathOperator{\argmin}{argmin}
\newcommand{\apxmed}{\textsc{ApxMed}}
\newcommand{\vecagg}{\textsc{VecAgg}}
\newcommand{\wei}{\textsc{wei}}
\newcommand{\ow}{\overline{w}}
\newcommand{\cD}{\mathcal{D}}
\newcommand{\maralg}{\textsc{TwoMarAlg}}
\newcommand{\ts}{\tilde{s}}
\newcommand{\clip}{\mathrm{clip}}
\newcommand{\tO}{\tilde{O}}
\newcommand{\opt}{\mathrm{OPT}}
\newcommand{\cM}{\mathcal{M}}
\newcommand{\cW}{\mathcal{W}}

\allowdisplaybreaks

\newif\ifaaaicameraready
\aaaicamerareadyfalse

\begin{document}

\maketitle

\begin{abstract}
Rank aggregation is a task of combining the rankings of items from multiple users into a single ranking 
that best represents the users' rankings. Alabi et al. (AAAI'22) presents differentially-private (DP) polynomial-time approximation schemes (PTASes) and $5$-approximation algorithms with certain additive errors for the \emph{Kemeny rank aggregation} problem in both central and local models.
In this paper, we present improved DP PTASes with smaller additive error in the central model. Furthermore, we are first to study the \emph{footrule rank aggregation} problem under DP. We give a near-optimal algorithm for this problem; as a corollary, this leads to 2-approximation algorithms with the same additive error as the $5$-approximation algorithms of Alabi et al. for the Kemeny rank aggregation problem in both central and local models.

\end{abstract}


\section{Introduction}

The \emph{rank aggregation} problem aims to combine $n$ individual rankings over $m$ candidates into a single consensus ranking that best reflects the input preferences. We consider two well-known optimality criteria for this task. 

The first is the \emph{Kemeny ranking}~\cite{Kemeny59}. Given a collection of rankings $\{\pi_1, \dots, \pi_n\}$, the \emph{Kendall's tau distance} between a permutation $\psi$ and $\pi_i$--denoted by $K(\psi, \pi_i)$--is defined\footnote{See \Cref{prelim:rank-agg} for formal definitions of rank aggregation.} as the number of candidate pairs $(j, k)$ for which the relative order in $\psi$ disagrees with that in $\pi_i$. The average Kendall's tau distance of $\psi$ is then $\frac{1}{n} \sum_{i=1}^n K(\psi, \pi_i)$. A permutation $\psi$ minimizing this quantity is called the \emph{Kemeny optimal ranking}, and the problem of finding the ranking is called the \emph{Kemeny rank aggregation problem}.

The second criterion is the \emph{footrule optimal ranking}, which minimizes the average position-wise discrepancy: $\frac{1}{n} \sum_{i=1}^n \sum_{j=1}^m |\psi(j) - \pi_i(j)|$, where $\pi_i(j)$ denotes the position of candidate $j$ in the $i$-th input ranking. 
The problem of finding the ranking is the \emph{footrule rank aggregation problem}.

While the \emph{footrule rank aggregation problem} can be solved in polynomial time \cite{DworkKNS01},  
the \textit{Kemeny rank aggregation problem} is NP-hard \cite{BartholdiTT89} even when the number of input rankings is as small as four \cite{DworkKNS01}. Thus, approximation algorithms for this problem have been studied. Several 2-approximation algorithms have been proposed \cite{DG77, DworkKNS01}.
Later, \citet{AilonCN08} introduced algorithm called \emph{KwikSort}, which also yields 2-approximation.
With a slight adjustment to the KwikSort algorithm, they  obtained an improved 11/7-approximation algorithm. Lastly, a polynomial-time approximation scheme (PTAS) for this problem was developed by \citet{MathieuS07}.

Often, the input rankings, e.g. votes in an election or preferences over search results, contain highly sensitive information about individuals. It is thus important that we perform rank aggregation while respecting their privacy.

In this work, we address both rank aggregation problems under the framework of differential privacy (DP). We recall the definition of DP below \cite{DworkMNS06,DworkKMMN06}.
\begin{definition} [$(\varepsilon, \delta)$-DP] \label{def:dp} Let $\varepsilon, \delta \geq 0$.
    A randomized mechanism $\mathcal{M}: \cX^n \rightarrow \cY$ is $(\varepsilon, \delta)$-DP if, for any \emph{neighboring}\footnote{We say that two datasets $\bX = (X_1, \dots, X_n), \bX' = (X'_1, \dots, X'_n)$ are neighbors if they differ on a single coordinate, i.e. there exists $i^* \in [n]$ such that $X_i = X'_i$ for all $i \in [n] \setminus \{i^*\}$.} datasets $\bX, \bX' \in \cX^n$ and any set of output $\cO \subseteq \cY$, 
    $\Pr[\mathcal{M}(\bX) \in \cO] \leq e^{\varepsilon} \cdot \Pr[\mathcal{M}(\bX') \in \cO] + \delta.$
\end{definition}
When $\delta = 0$, we say that the mechanism is $\eps$-DP; we refer to this case as \emph{pure-DP}, and the case $\delta > 0$ as \emph{approx-DP}.

Since its introduction, DP has become the gold standard for privacy-preserving data analysis \cite{abowd2018us,erlingsson2014rappor}, providing a mathematically rigorous means of quantifying privacy loss. 
  %
%
%
We also consider a variant of DP known as \emph{local DP} (LDP), in which each user independently perturbs their own data before transmission. This approach ensures privacy without requiring a trusted central server~\cite{kasiviswanathan2011can}.

Rank aggregation under DP have been the subject of multiple recent studies~\citep{HayEM17, YanLL20, LiuLXZ20,AlabiGKM22,XSC23,LanSLL24,XSG24}. A central challenge in designing DP algorithms—including those for rank aggregation—is managing the \emph{privacy-utility tradeoff}. In other words, the objective is to develop algorithms that, given a fixed privacy budget $(\eps, \delta)$, achieve the highest possible utility. In the context of rank aggregation, utility is typically measured by the Kendall's tau distance, with smaller distances indicating better utility. Prior work has addressed this challenge from both theoretical and empirical standpoints. 
We focus on the theoretical utility guarantee similar to \cite{HayEM17,AlabiGKM22}. Specifically, we say that a randomized algorithm is an $(\alpha, \beta)$-approximation algorithm if the expected cost of its output is at most $\alpha$ times the optimum plus $\beta$.


\subsection{Our Contributions}

\paragraph{Footrule Ranking} We present the \emph{first DP algorithm} for the footrule rank aggregation problem (Section~\ref{sec:reduction}). Our algorithm achieves the optimal approximation ratio $\alpha = 1$, while the additive errors are:
\begin{itemize}
    \item $\beta = \tilde{O}(m^3 / \varepsilon n)$ for $\eps$-DP,
    \item $\beta = \tilde{O}_\delta(m^{2.5} / \varepsilon n)$ for $(\eps,\delta)$-DP, 
    \item $\beta = \tilde{O}(m^{2.5} / \varepsilon \sqrt{n})$ for $\eps$-LDP.
\end{itemize}

Due to lower bounds from \cite{AlabiGKM22}, our additive errors for $\eps$-DP and $(\eps,\delta)$-DP are nearly optimal (up to logarithmic factors). 
\ifaaaicameraready \else
See \Cref{app:lb} for additional explanation.
\fi

\paragraph{Kemeny Ranking} Our contributions for the Kemeny rank aggregation problem are summarized in Table~\ref{table:results}. 

\begin{table*}[h]
\centering
\begin{tabular}{c c c l}
\hline
& $\alpha$ & $\beta$ & \\
\hline
$(\varepsilon,\delta$)-DP
& $1 + \xi$ & $\frac{1}{\varepsilon n} \tilde{O}_{\xi,\delta}(m^3)$ & \cite{AlabiGKM22} \\
& $5 + \xi$ & $\frac{1}{\varepsilon n} \tilde{O}_{\xi,\delta}(m^{2.5})$ & \cite{AlabiGKM22} \\
& any $\alpha \geq 1$ & $\frac{1}{\varepsilon n} \Omega(m^{2.5})$ & \cite{AlabiGKM22} [Lower Bound]  \vspace{0.1cm} \\ \vspace{0.1cm}
& \textbf{$1 + \xi$} & \textbf{$\frac{1}{\varepsilon n} \tilde{O}_{\xi,\delta}(m^{65/22})$} & \textbf{Our Work [\Cref{sec:PTAS}]} \\ \vspace{0.1cm}
& \textbf{$2$} & \textbf{$\frac{1}{\varepsilon n} \tilde{O}_\delta(m^{2.5})$} & \textbf{Our Work  [\Cref{sec:reduction}]} \vspace{0.1cm} \\ 
\hline
$\varepsilon$-DP & $1$ & $\frac{1}{\varepsilon n} \tilde{O}(m^3)$ 
& \cite{HayEM17} \\
& $1 + \xi$ & $\frac{1}{\varepsilon n} \tilde{O}_\xi(m^4)$ & \cite{AlabiGKM22} \\
& $5 + \xi$ & $\frac{1}{\varepsilon n} \tilde{O}_\xi(m^3)$ & \cite{AlabiGKM22} \vspace{0.1cm} \\
& any $\alpha \geq 1$ & $\frac{1}{\varepsilon n} \Omega(m^{3})$ & \cite{AlabiGKM22} [Lower Bound] \vspace{0.1cm} \\ 
& \textbf{$1 + \xi$} & \textbf{$\frac{1}{\varepsilon n} \tilde{O}_{\xi}(m^{27/7})$} & \textbf{Our Work [\Cref{sec:PTAS}]} \\ 
& \textbf{$2$} & \textbf{$\frac{1}{\varepsilon n} \tilde{O}(m^3)$} & \textbf{Our Work  [\Cref{sec:reduction}]} \vspace{0.1cm} \\
\hline
$\varepsilon$-LDP & $1 + \xi$ & $\frac{1}{\varepsilon \sqrt{n}} \tilde{O}_{\xi}(m^3)$ & \cite{AlabiGKM22} \\
& $5 + \xi$ & $\frac{1}{\varepsilon \sqrt{n}} \tilde{O}_{\xi}(m^{2.5})$ & \cite{AlabiGKM22} \vspace{0.1cm} \\
\vspace{0.1cm}
& \textbf{$2$} & \textbf{$\frac{1}{\varepsilon \sqrt{n}} \tilde{O}(m^{2.5})$} & \textbf{Our Work  [\Cref{sec:reduction}]} \\
\hline
\end{tabular}
\caption{Comparison of 
our $(\alpha, \beta)$-approximation algorithms for Kemeny Rank Aggregation with previous works. All algorithms run in $(nm)^{O(1)}$ time, except the exponential mechanism \cite{HayEM17} which runs in $O(n \cdot m!)$ time.}
\label{table:results}
\end{table*}

The state-of-the-art work by~\citet{AlabiGKM22} provides two $(\alpha, \beta)$-approximation algorithms for the Kemeny rank aggregation problem under $(\varepsilon, \delta)$-differential privacy. The first algorithm is a PTAS, which achieves a small multiplicative factor $\alpha = 1 + \xi$ for $\xi > 0$ but incurs a large additive error $\beta = \tilde{O}_{\xi,\delta}(m^3)/(\varepsilon n)$. In contrast, the second algorithm offers a smaller additive term $\beta = \tilde{O}_{\delta}(m^{2.5})/(\varepsilon n)$, at the cost of a larger multiplicative factor $\alpha = 5 + \xi$. It is demonstrated in the paper that there is no algorithm with additive term $\beta = o(m^{2.5})/(\varepsilon n)$.

We aim to improve upon these trade-offs by designing an algorithm that retains the small $\alpha$ of the first while reducing the additive error $\beta$, or alternatively, one that preserves the small $\beta$ of the second while lowering the multiplicative factor $\alpha$. 

Our first algorithm, presented as Algorithm~\ref{alg:footrule-rank-agg} in Section~\ref{sec:reduction}, improves the multiplicative factor $\alpha$ of the second algorithm in~\cite{AlabiGKM22} from $5 + \xi$ to $2$, while preserving the same additive error $\beta$.
 Moreover, this approach extends naturally to the $\varepsilon$-DP and  $\varepsilon$-LDP settings, achieving the same reduced multiplicative factor $\alpha = 2$ without increasing the additive error $\beta$. For LDP, our algorithm also has the advantage of being \emph{non-interactive} whereas Alabi et al.'s algorithm is \emph{interactive} and requires $O(\log m)$ rounds of communication.

Our second algorithm, presented as Algorithms \ref{alg:ptas_small}-\ref{alg:ptas_large} in Section~\ref{sec:PTAS}, improves the additive error $\beta$ of the first algorithm in \cite{AlabiGKM22} from $\frac{1}{\varepsilon n} \tilde{O}_{\varepsilon,\delta}(m^3)$ to $\frac{1}{\varepsilon n} \tilde{O}_{\varepsilon,\delta}(m^{65/22})$. As it has been shown that no algorithm can achieve an additive term of order $o(m^{2.5})/(\varepsilon n)$~\cite{AlabiGKM22}, we believe that reducing the exponent of $m$ to a value strictly less than $3$ represents a meaningful advancement toward closing the gap between the known upper and lower bounds. 
We also achieve a similar improvement for $\eps$-DP.

It is worth noting that a $(1, \tilde{O}(m^3 / \varepsilon n))$-approximation algorithm can be obtained using the exponential mechanism. However, as observed in~\cite{HayEM17}, this approach is not computationally efficient, as it does not run in polynomial time. 
\ifaaaicameraready \else
A comparison between our algorithm and two additional algorithms from the same paper is provided in Appendix~\ref{app:comparison}.
\fi

\label{sec:contribution}

\subsection{Technical Overview}
\label{subsec:overview}

We provide high-level technical overview of our algorithms. 
\ifaaaicameraready \else
For brevity, we will sometimes be informal; all details will be formalized in the corresponding sections.
\fi

\paragraph{Algorithms for Footrule Ranking} Our algorithm for the footrule ranking builds upon the non-private approach of~\cite{DworkKNS01}. For each pair \( q, j \in [m] \), let the cost of assigning candidate \( q \) to position \( j \) be \( \gamma_q^j = \frac{1}{n} \sum_{i = 1}^{n} |\pi_i(q) - j| \). Given a permutation \( \psi \), its total cost is \( \sum_{q \in [m]} \gamma_q^{\psi(q)} \). The objective is thus to find 
\( \psi \) that minimizes this cost. This corresponds to a minimum weight bipartite matching problem, which can be solved efficiently.

To construct a DP version of the algorithm, we propose a method to accurately release the weights $\gamma_q^j$ for all $q$ and $j$ privately. If we can publish each value of $\gamma_q^j$ with additive error of $\beta$, since there are $m$ values of $\gamma_q^j$ in the cost calculation, we obtain the additive term of $m \cdot \beta$ in the cost.

In a naive approach, each value $\pi_i(q)$ contributes to up to $m$ output values (i.e. $\gamma_q^j$ for all $j \in [m]$), causing the privacy budget to be split across many computations. This leads to large noise and significant error. To mitigate this, we adapt the binary tree mechanism~\cite{DworkNPR10,ChanSS11} to better manage the privacy budget.

The original binary tree mechanism is designed to answer multiple range queries, specifically the frequencies of data within intervals $[low, up]$. To achieve this, we define a set of intervals $\mathcal{I} = \{[(p - 1) \cdot 2^\ell + 1, p \cdot 2^\ell] : \ell \in [\lceil \log m \rceil], p \in [\lceil m / 2^\ell \rceil]\}$. We privately release the frequency of data falling within each interval in $\mathcal{I}$. Since any query range $[low, up]$ can be decomposed into a disjoint union of intervals from $\mathcal{I}$, we can answer the query by aggregating the corresponding private frequencies. In this mechanism, each data point is used in only $O(\log m)$ intervals, which is significantly fewer than in the naive approach. As a result, the added noise is smaller, leading to a reduced additive error.

Let $r(I)$ denote the smallest value in the interval $I$. To compute $\gamma_q^j$ for all $q$ and $j$, we privately release, for each interval $I \in \mathcal{I}$, the values $\frac{1}{n} \sum_{i : \pi_i(q) \in I} (\pi_i(q) - r(I))$ and $|\{i : \pi_i(q) \in I\}|$. These allow us to calculate $\frac{1}{n} \sum_{i : \pi_i(q) \in I} |\pi_i(q) - j|$ when $j \notin I$. Since there exists a collection of disjoint intervals in $\mathcal{I}$ whose union equals $[m] \setminus \{j\}$, summing over these intervals yields $\frac{1}{n} \sum_{i=1}^n |\pi_i(q) - j|$. Similar to the original binary tree mechanism, each data point appears in only $O(\log m)$ intervals instead of $m$, resulting in a smaller additive error in our mechanism.

\paragraph{$(2,\beta)$-approximation algorithm for the Kemeny rank aggregation} Since the Spearman's footrule distance is at least the Kendall's tau distance and at most twice that distance, any $(1, \beta)$-approximation algorithm for the footrule rank aggregation, as described in the previous paragraph, also serves as a $(2, \beta)$-approximation algorithm for the Kemeny rank aggregation problem.

\paragraph{$(1+\xi,\beta)$-approximation algorithm for the Kemeny rank aggregation}

Our improved PTAS requires separate treatments of the cases where $n$ is ``small'' and where $n$ is ``large''. For exposition purpose, we focus on $(\eps,\delta)$-DP and it is best to think of $n$ as being large when it is slightly 
larger than $m$, i.e. $n = \Theta_{\eps, \delta}(m \log m)$.  Note that, although this is the regime relevant to most practical ranking applications—where the number of voters $n$ is typically much larger than the number of candidates $m$—the error guaranteed by the PTAS in \cite{AlabiGKM22} is only $\tilde{O}_{\eps,\delta}\!\left(m^3\right)$, which is barely nontrivial, since any ranking incurs error at most $m^3$.
\ifaaaicameraready \else
 Below we show how to (significantly) improve on this when $n$ is small and when $n$ is large.
\fi

\paragraph{Large $n$: Leveraging Unbiasedness} Define the matrix $\bw^{\Pi} \in [0, 1]^{m \times m}$ by 
$w^{\Pi}_{uv} := \frac{1}{n} \sum_{i \in [n]} \bone[\pi_i(u) < \pi_i(v)]$.
We start from the following observation: 
we can privatize $\bw^{\Pi}$ by adding independent Gaussian noise $\cN(0, \sigma^2)$ where $\sigma = O\left(\frac{m \log(1/\delta)}{\eps n}\right)$ to each of its entry. Let us call the privatized matrix $\tbw$. Indeed, it is simple to see that, if we look at any permutation $\pi$, the difference in its cost between $\bw^{\Pi}$ and $\tbw$--which is a summation of $O(m^2)$ of the Gaussian noise--has variance only $O_{\eps, \delta}(m^2)$. Since there are only $m!$ rankings in total, a union bound (and concentration of Gaussian) then tells us that w.h.p. the cost difference for \emph{every} ranking is at most $O_{\eps, \delta}(m^{2.5})$. So if we run the non-private PTAS~\cite{MathieuS07} on $\tbw$, then we should already improve upon~\cite{AlabiGKM22}, right?
Although ostensibly correct, there is one flaw in this argument: $\tbw$ can have negative entries, for which the non-private PTAS cannot handle. Indeed, if we proceed by ``clipping'' $\tbw$ to ensure non-negativity, then the bias would make the error become $O_{\eps, \delta}(m^3)$, so unfortunately we obtain no improvement over \cite{AlabiGKM22}.

To overcome this issue, our main observation is that, to get a good approximation for $\bw^{\Pi}$, it suffices to get a good approximation for the instance where, for some $u, v$ such that $w_{uv} < w_{vu}$, we replace $w_{uv}$ with zero and replace $w_{vu}$ with $w_{vu} - w_{uv}$. Using this observation, by adding the noise only to the latter, we can ensure non-negativity with high probability as long as $\sigma = O\left(\frac{1}{\log m}\right)$. With some additional technical work, this solves the bias issue. However, the requirement on $\sigma$ means that $n$ has to be at least $\Omega_{\eps, \delta}(m \log m)$. This indeed necessitates another algorithm when $n$ is small.

\paragraph{Small $n$: Reduction to 2-Way Marginal.}
In the case of small $n$, our intuition starts from the following: Suppose instead of each $\pi_i$ being a permutation $[m] \to [m]$, it is a mapping $[m] \to \{0, 1\}$. Then, the term $\bone[\pi_i(u) < \pi_i(v)]$ is exactly equal to $\bone[\pi_i(u) = 0, \pi_i(v) = 1]$. Thus, each entry of $\bw^{\Pi}$ is simply a \emph{2-way marginal} query, which is well studied in DP literature. In particular, there is an efficient algorithm that can answer such queries with error $O_{\eps, \delta}(m^{1/4}/\sqrt{n})$ \cite{DworkNT15} per query. If we were to achieve this error for estimating $\bw^{\Pi}$, then we would already achieve improvement over \cite{AlabiGKM22} for $n = O_{\eps, \delta}(m \log m)$.

The main challenge is now to extend the above simplified setting $\pi_i: [m] \to \{0, 1\}$ to the desired setting $\pi_i: [m] \to [m]$. To do this, we apply \emph{bucketing}. 
Roughly speaking, we pick a number of buckets $B$, divide the range $[m]$ into $B$ buckets of equal size and separate the term $\bone[\pi_i(u) < \pi_i(v)]$ into two parts, based on whether $\pi_i(u), \pi_i(v)$ are from the same bucket. If they are from the same bucket, we use the Gaussian mechanism; since there are only $m^2/B$ such pairs, we can add smaller noise than without bucketing. If they are from different buckets, then we use a generalization of the 2-way marginal reduction described in the previous paragraph. By picking $B$ appropriately, we can optimize the error which ends up being $\tilde{O}_{\eps, \delta}(m^{43/22})$ when $n = \tilde{\Theta}_{\eps, \delta}(m)$.

\section{Preliminaries}

\subsection{Rank Aggregation}
\label{prelim:rank-agg}

For any $m \in \N$, we use $[m]$ to denote $\{1, \dots, m\}$ and let $\bbS_m$ denote the set of all permutations (i.e., all rankings) on $[m]$. In the rank aggregation problem, we are given a dataset $\Pi$ of $n$ rankings $(\pi_1, \dots, \pi_n) \in (\bbS_m)^n$, the goal is to output a ranking $\psi \in \bbS_m$ that minimizes $d(\psi, \Pi) := \frac{1}{n} \sum_{i \in [n]} d(\psi, \pi_i)$ where $d: \bbS_m \times \bbS_m \to \R_{\geq 0}$ is a certain distance. We will discuss two distances here:
\begin{itemize}
\item Spearman's footrule distance: \\$F(\psi, \pi) := \sum_{j \in [m]} |\psi(j) - \pi(j)|$.
\item  Kendall's tau distance: $K(\psi, \pi) := |\{(j, j') \in [m] \times [m] \mid \pi(j) < \pi(j') \wedge \psi(j) > \psi(j')\}|$.
\end{itemize}

An algorithm is \emph{efficient} if it runs in $(nm)^{O(1)}$ time.

\subsection{Differential Privacy}

We use the DP notion as defined in \Cref{def:dp}. Note that in rank aggregation problems, we have $\cX = \cY = \bbS_m$, and two datasets $\Pi, \Pi' \in (\bbS_m)^n$ are neighbors iff they differ on a single ranking. We also consider the non-interactive local model, defined below.

\begin{definition}[$\varepsilon$-LDP]
An algorithm in the (non-interactive) local model consists of a randomizer $\cR: \cX \to \cZ$ and an analyst $\cA: \cZ^n \to \cY$; with input dataset $\bX = (x_1, \dots, x_n)$, the final output is computed as $\cA(\cR(x_1), \dots, \cR(x_n))$. The algorithm (or the randomizer) is said to be $\eps$-LDP iff, for any $x, x' \in \cX$ and $\cO \subseteq \cZ$, $$\Pr[\cR(x) \in \cO] \leq e^\eps \cdot \Pr[\cR(x') \in \cO].$$
\end{definition}

Note that these DP notions are the same as the ones used in \cite{HayEM17,AlabiGKM22}.

For simplicity of presentation, we will assume throughout that $\eps \leq 1, \delta \in (0, 1/2)$, and we will not state this explicitly.

\ifaaaicameraready \else
Additional preliminaries (including sensitivity and basic DP mechanisms) can be found in \Cref{app:add-prelim}. Due to space constraints, we also defer most full proofs to the appendix.
\fi

\section{DP Approximate Median}
\label{sec:median}
In this section, we study the Approximate Median~(\apxmed) problem, which is closely related to the footrule ranking objective.  
The universe is $\cX = [m]$, and the input is $\bx = (x_1, \dots, x_n) \in [m]^n$.  
For each $j \in [m]$, define
  $\gamma_j(\bx) := \frac{1}{n} \sum_{i=1}^n \lvert x_i - j \rvert$.
The goal is to output estimates $(\tgamma_1, \dots, \tgamma_m)$.  
We measure the quality of these estimates by the $\ell_\infty$-error
  $\max_{j \in [m]} \lvert \tgamma_j - \gamma_j(\bx) \rvert$.
We say that an algorithm is $\beta$-accurate if the expected $\ell_\infty$-error of its output is at most $\beta$.
%

We refer to this task as Approximate Median because the median $j^*$ is a minimizer to $\gamma_{j^*}(\bx)$ and, thus, answering the queries $(\gamma_{j}(\bx))_{j \in [m]}$ allows us to determine how ``far'' each $j$ is from the median. While other approximate notions of median have been studied under DP (cf. \cite{GillenwaterJK21}), to the best of our knowledge, \apxmed~has not been considered  and may be of independent interest beyond the context of rank aggregation, e.g. in DP clustering~\cite{GhaziKM20}.

\ifaaaicameraready \else
As alluded to in \Cref{subsec:overview}, a key ingredient of our improved rank aggregation algorithms is a DP algorithm for \apxmed, which is stated more formally below.
\fi

\begin{theorem} \label{thm:apx-median}
There is an efficient algorithm for \apxmed~that is $\beta$-accurate with
\begin{itemize}
\item $\beta = O\left(\frac{m \log m}{\eps n}\right)$ for $\eps$-DP,
\item $\beta = O\left(\frac{m \sqrt{\log m \log(1/\delta)}}{\eps n}\right)$ for $(\eps, \delta)$-DP,
\item $\beta = O\left(\frac{m \sqrt{\log m}}{\eps \sqrt{n}}\right)$ for $\eps$-LDP.
\end{itemize}
\end{theorem}


\paragraph{Modified Binary Tree Mechanism}

We present an algorithm for answering approximate median queries by adapting the well-known \emph{Binary Tree Mechanism}~\cite{DworkNPR10,ChanSS11}, originally designed for range queries. We assume wlog\footnote{If not, we may increase $m$ to the next power of two.} that the domain size $m$ is a power of two, i.e., $m = 2^d$ for some $d \in \N$. 

We define a complete binary tree with the set of 
nodes $\cB$; each node $t \in \cB$ is associated with a subset of the domain, denoted by $I(t) \subseteq [m]$. Each leaf corresponds to a singleton set. The tree has $d + 1$ levels, indexed by $\ell \in \{0, \dots, d\}$ (where leaves are at level 0). For any internal node $t$ at level $\ell$, the set $I(t)$ consists of the union of $I(t')$ for all leaves $t'$ in its subtree and is of the form $I(t) = [(p - 1) \cdot 2^\ell + 1, p \cdot 2^\ell]$ for some $p \in [m / 2^\ell]$. A node $t$ in the tree is called a \textit{left node} if all elements of $I(t)$ are smaller than those of its sibling. 
The smallest value in $I(t)$ is denoted by $r(t)$ and the level of the node $t$ is defined by $\ell(t)$.

Our algorithm is presented in \Cref{algo:modified-binary}.
In the original binary tree mechanism~\cite{DworkNPR10,ChanSS11}, each node $t$ stores the number of the data points that fall within its associated interval $I(t)$. In our algorithm, to support the \apxmed~computation, each node $t$ maintains $v_t^{\rm agg} = \frac{1}{n} \sum_{x_i \in I(t)} |x_i - r(t)|$ and $u_t^{\rm agg} = \frac{1}{n} \sum_{x_i \in I(t)} 2^{\ell(t)}$.

For \apxmed, we consider all sibling nodes of the nodes whose intervals contain the candidate value $j$. These sibling intervals are disjoint and their union is equal to $[m] \setminus \{j\}$. Therefore, Lines \ref{line:start-computation}–\ref{line:finish-estimate} of~\Cref{algo:modified-binary} compute the contribution of all data points $x_i \ne j$ exactly once. This allows us to estimate the quantity $\frac{1}{n} \sum_i |x_i - j|$ by the expression $\frac{1}{n} \sum_{t' \in \cB'_j} \sum_{x_i \in I(t')} |x_i - j|$, where $\cB'_j$ denotes the set of relevant sibling nodes. The sum $\sum_{x_i \in I(t')} |x_i - j|$ is estimated by the computation at Line \ref{line:finish-estimate} of the algorithm.

Note that we have to compute $v_t^{\rm agg}, u_t^{\rm agg}$ while respecting the privacy constraints. To do this, we write out these as summands $v_t^{\rm agg} = \frac{1}{n} \sum_{i \in [n]} \bone[x_i \in I(t)] \cdot (x_i - r(t))$ and  $u_t^{\rm agg} = \frac{1}{n} \sum_{i \in [n]} \bone[x_i \in I(t)] \cdot 2^{\ell(t)}$. Note that each inner term depends only on $x_i$. This means that $v^{\rm agg}, u^{\rm agg}$ are simply the average of $n$ vectors, where each vector depends only on a single $x_i$. Thus, we may employ known \emph{vector aggregation algorithms} in $\eps$-DP, $(\eps,\delta)$-DP and $\eps$-LDP for this task; indeed, this is the only difference between the three settings. Specifically, we use the Laplace mechanism, the Gaussian mechanism, and an algorithm of \cite{duchi2014localprivacydataprocessing}, respectively. Finally, to optimize the error further, we employ a weighting strategy where we multiply the contributions to $v_t^{\rm agg}, u_t^{\rm agg}$ by $\kappa^{d-\ell(t)}$ before aggregation (Lines \ref{line:wei-firstline}-\ref{line:wei-secondline}), and multiplying $\kappa^{\ell(t)-d}$ back to the estimates (Lines \ref{line:rewei-firstline}-\ref{line:rewei-secondline}). Here $\kappa$ can be any constant between 1 and 2. This helps reduce the error by a logarithmic factor.

\ifaaaicameraready \else
We defer the full proof of \Cref{thm:apx-median} to \Cref{app:median}.
\fi

\begin{algorithm}
    \caption{Modified binary tree mechanism for \apxmed}
    \label{algo:modified-binary}
\textbf{Input:} Dataset $\bx = (x_1, \dots, x_n) \in [m]^n$\\
\textbf{Output:} Vector $(\tgamma_j)_{j \in [m]}$ \\ 
\textbf{Parameters: } DP vector aggregation algorithm \vecagg, weight growth factor $\kappa \in (1, 2)$
\begin{algorithmic}[1]
\FORALL[Weighted Contribution]{$t \in \cB, i \in [n]$} 
\STATE $v_{i, t}^{\wei} \gets \kappa^{d - \ell(t)} \cdot \bone[x_i \in I(t)] \cdot (x_i - r(t))$ \label{line:wei-firstline}
\STATE $u_{i, t}^{\wei} \gets \kappa^{d - \ell(t)} \cdot \bone[x_i \in I(t)] \cdot 2^{\ell(t)}$ \label{line:wei-secondline}
\ENDFOR

\STATE Use $\vecagg$ to aggregate vectors $v_{i}^{\wei}$ and $u_i^{\wei}$ across all $i \in [n]$; let $\tv^{\agg, \wei}$ and $\tu^{\agg, \wei}$ be the result. \label{line:vec-agg}

\FORALL[Reweight Aggregate Vector]{$t \in \cB$} 
\STATE $\tv_{t}^{\agg} \gets \kappa^{\ell(t) - d} \cdot \tv_{t}^{\agg, \wei}$ \label{line:rewei-firstline}
\STATE $\tu_{t}^{\agg} \gets \kappa^{\ell(t) - d} \cdot \tu_{t}^{\agg, \wei}$ \label{line:rewei-secondline}
\ENDFOR

        \FORALL[Compute Estimates]{$j \in [m]$} \label{line:start-computation}
            \STATE $\cB_j \gets$ the set of non-root nodes $t$ such that $j \in I(t)$. 
            \FORALL{$t \in \cB_j$}
            \STATE $t' \gets $ sibling of $t$
            \STATE $s_{j,t} \gets \begin{cases} 1 &\text{ if } t \text{ is a left node} \\ -1 &\text{ otherwise}.\end{cases}$
            \STATE $\tgamma_{j,t} \gets s_{j, t} \cdot \left(\tv^\agg_{t'} + \frac{r(t') - j}{2^{\ell(t')}} \cdot \tu^{\agg}_{t'}\right)$. \label{line:finish-estimate}
            \ENDFOR
            \STATE $\tgamma_j \gets \sum_{t \in \cB_j} \tgamma_{j,t}$  \label{line:estimate}
        \ENDFOR
        \RETURN $(\tgamma_j)_{j \in [m]}$
    \end{algorithmic}
\end{algorithm}

\paragraph{Parallel Approximate Median.}
Now, we consider an \emph{$m$-parallel} version of $\apxmed$, which will be convenient for the next application. In this variant, the domain $\cX$ is now $[m]^m$; i.e. the $i$-th input is now $x_i = (x_{i,1}, \dots, x_{i,m})$. The goal is to compute approximate median for $x_{1,q}, \dots, x_{n,q}$ for all $q \in [m]$. That is, for all $j \in [m], q \in [m]$, we wish to estimate $\gamma_{j,q}(\bx) := \frac{1}{n} \sum_{i=1}^n |x_{i,q} - j|$. Similar to before, the $\ell_\infty$-error of the estimates $(\tgamma_{j,q})_{j,q\in[m]}$ is defined as $\max_{j,q \in [m]} |\tgamma_{j,q} - \gamma_{j,q}(\bx)|$. An algorithm is $\beta$-accurate if the expected $\ell_\infty$-error of its output is at most $\beta$. 

The following is a simple corollary of \Cref{thm:apx-median}:

\begin{corollary} \label{cor:par-apx-median}
There is an efficient algorithm for $m$-parallel \apxmed~that is $\beta$-accurate with
\begin{itemize}
\item $\beta = O\left(\frac{m^2 \log m}{\eps n}\right)$ for $\eps$-DP,
\item $\beta = O\left(\frac{m^{1.5} \sqrt{\log m \log(1/\delta)}}{\eps n}\right)$ for $(\eps, \delta)$-DP,
\item $\beta = O\left(\frac{m^{1.5} \sqrt{\log m}}{\eps \sqrt{n}}\right)$ for $\eps$-LDP.
\end{itemize}
\end{corollary}

\section{From \apxmed~to Rank Aggregation}
\label{sec:reduction}

We next provide a reduction from \apxmed~to footrule rank aggregation. We use the matching-based non-private algorithm of \cite{DworkKNS01} but with the weights computed based on our \apxmed~algorithm, as shown in \Cref{alg:footrule-rank-agg}.

\begin{algorithm}
    \caption{Footrule Rank Aggregation from \apxmed}
    \label{alg:footrule-rank-agg}
\textbf{Input:} Dataset $\Pi = (\pi_1, \dots, \pi_n) \in (\bbS_m)^n$\\
\textbf{Output:} Aggregated rank $\psi \in \bbS_m$ 
\begin{algorithmic}[1]
\STATE $(\tgamma_{j, q})_{j, q \in [m]} \gets$ output from running \emph{$m$-parallel \apxmed}~algorithm in \Cref{cor:par-apx-median} on input $\Pi$.

\STATE $G \gets$ weighted complete bipartite graph where both the left and right vertex sets are $[m]$, and the weight of each edge $(q, j)$ is $\tgamma_{j,q}$.
\RETURN Minimum-weight bipartite matching of $G$. (i.e., if edge $(a, b)$ belongs to the matching, let $\psi(a) = b$.)
\end{algorithmic}
\end{algorithm}

We can show that the additive error of \Cref{alg:footrule-rank-agg} is $O(m)$ times the $\ell_\infty$-error of the estimates for $m$-parallel \apxmed. Thus, from \Cref{cor:par-apx-median}, we immediately have

\begin{theorem}
    \label{thm:red-median-to-rank-agg}
There is an efficient $(1, \beta)$-approximation algorithm for footrule rank aggregation with
\begin{itemize}
\item $\beta = O\left(\frac{m^3 \log m}{\eps n}\right)$ for $\eps$-DP,
\item $\beta = O\left(\frac{m^{2.5} \sqrt{\log m \log(1/\delta)}}{\eps n}\right)$ for $(\eps, \delta)$-DP,
\item $\beta = O\left(\frac{m^{2.5} \sqrt{\log m}}{\eps \sqrt{n}}\right)$ for $\eps$-LDP.
\end{itemize}
\end{theorem}


Since the Spearman's footrule distance is always at least the Kendall's tau distance and at most twice that distance, we also get the following as a corollary:

\begin{theorem} \label{thm:two-apx}
There is an efficient $(2, \beta)$-approximation algorithm for Kemeny rank aggregation with
\begin{itemize}
\item $\beta = O\left(\frac{m^3 \log m}{\eps n}\right)$ for $\eps$-DP,
\item $\beta = O\left(\frac{m^{2.5} \sqrt{\log m \log(1/\delta)}}{\eps n}\right)$ for $(\eps, \delta)$-DP,
\item $\beta = O\left(\frac{m^{2.5} \sqrt{\log m}}{\eps \sqrt{n}}\right)$ for $\eps$-LDP.
\end{itemize}
\end{theorem}

As mentioned earlier, the additive errors $\beta$ in \Cref{thm:red-median-to-rank-agg,thm:two-apx} for $\eps$-DP and $(\eps,\delta)$-DP are optimal due to the lower bound of \cite{AlabiGKM22}\ifaaaicameraready.\else; see \Cref{app:lb} for a more detailed discussion.\fi

\section{PTAS for Kemeny Rank Aggregation}
\label{sec:PTAS}

In this section, we present our PTAS for Kemeny rank aggregation. Recall that \citet{AlabiGKM22} gave PTASes for the problem with additive errors $\beta = \tilde{O}_{\xi} \left(\frac{m^4}{\eps n}\right)$ and $\tilde{O}_{\xi,\delta}\left(\frac{m^3}{\eps n}\right)$ for $\eps$-DP and $(\eps, \delta)$-DP, respectively. Our algorithm  improves on these additive errors, as stated below.

\begin{theorem} \label{thm:ptas-merged}
For every constant $\xi > 0$, there is an efficient $(1 + \xi, \beta)$-approximation algorithm for Kemeny rank aggregation with
\begin{itemize}
\item $\beta = \tilde{O}_\xi\left(\frac{m^{27/7}}{\eps n}\right)$ for $\eps$-DP,
\item $\beta = \tilde{O}_\xi\left(\frac{m^{65/22} \sqrt{\log(1/\delta)}}{\eps n}\right)$ for $(\eps, \delta)$-DP,
\end{itemize}
\end{theorem}

\subsection{Additional Preliminaries}

Our PTAS will require additional tools and preliminaries, which we list below.

For an instance $\Pi = (\pi_1, \dots, \pi_n)$ of rank aggregation, we define the matrix $\bw^{\Pi} \in [0, 1]^{m \times m}$ by
\begin{align} \label{eq:def-pairwise-matrix}
w^{\Pi}_{uv} := \frac{1}{n} \sum_{i \in [n]} \bone[\pi_i(u) < \pi_i(v)]. 
\end{align}

One important observation to make is that the desired objective, $K(\psi, \Pi)$, can be written in terms of $\bw^{\Pi}$ as 
$K(\psi, \Pi) = \sum_{u, v \in [n] \atop \psi(u) < \psi(v)} w_{vu}^{\Pi}$.

\paragraph{Weighted Feedback Arc Set.} For our PTAS, we need to consider a more general problem, called the \emph{Weighted Feedback Arc Set (WFAS)} problem.

\begin{definition}[Weighted Feedback Arc Set (WFAS)]
The input is a set of candidates $[m]$ and a weight matrix $\bw \in \R^{m \times m}$. The goal is to output $\pi$ that minimizes $\cost_{\bw}(\pi) := \sum_{u,v\in [n] \atop \pi(u) < \pi(v)} w_{vu}.$
\end{definition}

Note that the Kemeny rank aggregation problem is a special case of WFAS where $\bw = \bw^{\Pi}$ is as defined in~\eqref{eq:def-pairwise-matrix}.
We remark that we deliberately allow the weights to be negative and unbounded as this will be useful later on. Nevertheless, we also need a boundedness definition here:
\begin{definition}[Bounded WFAS Instance] \label{def:bounded}
An instance $\bw$ is said to be \emph{bounded} if $w_{uv} \geq 0$ and $w_{uv} + w_{vu} \in \left[\frac{1}{2}, 2\right]$. 
\end{definition}

While WFAS is hard to approximate on general instances~\cite{GuruswamiHMRC11,MathieuS07} gave a PTAS for bounded instances\footnote{In fact, the algorithm of \cite{MathieuS07} allow the numbers 1/2, 2 in the boundedness assumption to be changed to any constants. However, we do not use this in our work.}:

\begin{theorem}[\citealt{MathieuS07}] \label{thm:MS}
For every constant $\xi > 0$, there exists an efficient $(1 + \xi)$-approximation algorithm for WFAS on bounded instances.
\end{theorem}


We also state another simple but important lemma below, which states that, if we give a good approximation algorithm on one weight matrix $\tbw$, it remains a good approximation on a ``nearby'' weight matrix $\bw$.

\begin{lemma}[\citealt{AlabiGKM22}, Theorem 1] \label{lem:err-to-apx}
For any $\tbw, \bw \in [0,1]^{m \times m}$ such that $\|\tbw - \bw\|_1 \leq e$, any $(1 + \xi)$-approximate solution to WFAS on $\tbw$ is also an $(1 + \xi, O(e))$-approximate solution to WFAS on $\bw$.
\end{lemma}

\paragraph{Two-Way Marginals.} Finally, we recall the two-way marginal problem which is defined as follows:
\begin{definition}[Two-Way Marginal Queries] \label{def:2-marginal}
The universe $\cX$ here is $\{0, 1\}^d$, the set of all binary vectors of length $d$.
For notational convenience, we view $x_i \in \{0, 1\}^d$ as a function $x_i: [d] \to \{0, 1\}$ instead. A two-way marginal query\footnote{Here we restrict two-way marginal queries only to those with attributes equal to ``1'', which is sufficient for our setting.} is indexed by a size-2 set $S = \{j_1, j_2\} \subseteq [d]$. On input dataset $\bX = (x_1, \dots, x_n)$, the query has value $q_{S}(\bX) = \frac{1}{n} \sum_{i \in [n]} \bone[x_i(j_1) = 1 \wedge x_i(j_2) = 1].$
The goal of an algorithm for two-way marginal is to output estimates $(\tq_S)_{S \in \binom{[d]}{2}}$ for the above queries.
\end{definition}

The two-way marginal problem is well-studied in DP literature (e.g.~\cite{BarakCDKMT07,BunUV14,DworkNT15,Nikolov24}). A crucial aspect we will use here is that when $n$ is small, there are efficient $\eps$-DP and $(\eps,\delta)$-DP algorithms \cite{Nikolov24,DworkNT15} that achieve significantly smaller errors compared to standard noise addition algorithms, i.e. the Laplace or Gaussian mechanisms.

As discussed in \Cref{subsec:overview}, our algorithm requires us to handle two cases based on whether $n$ is ``small'' or ``large''.

\subsection{Small $n$: Two-Way Marginal}
\Cref{lem:err-to-apx} gives us a fairly clear overall strategy towards obtaining good approximation algorithms for Kemeny rank aggregation: Design a differentially private algorithm that publishes the weight matrix $\bw^{\Pi}$ with small error. Indeed, this was the same strategy deployed in \cite{AlabiGKM22}, who use Laplace and Gaussian mechanisms to add noise to $\bw^{\Pi}$. To improve upon this in the small $n$ regime, we will instead employ DP two-way marginal algorithms. 
\ifaaaicameraready \else
As alluded to in \Cref{subsec:overview}, we will use bucketing for this purpose.
\fi

Let $B \in \N$ be a parameter to be specified. We partition $[m]$ into $B$ \emph{buckets} (indexed by $1, \dots, B$), where each bucket is an interval of size $\leq \lceil m/B \rceil$. For every $i \in [m]$, let $\iota(i)$ denote the index of the bucket it belongs to. 

The matrix $\bw^{\Pi}$ can now be decomposed into the sum of two parts: (i) comparison within the same bucket $\bs$, and (ii) comparison across buckets $\bt$.
More formally, we define the matrices $\bs, \bt \in \R^{m \times m}$ as follows:
\begin{align*}
s_{uv} &= \frac{1}{n} \sum_{i \in [n]} \bone[\iota(\pi_i(u)) = \iota(\pi_i(v)) \wedge \pi_i(u) < \pi_i(v)], \\
t_{uv} &= \frac{1}{n} \sum_{i \in [n]} \bone[\iota(\pi_i(u)) < \iota(\pi_i(v))].
\end{align*}
It is simple to see that $\bw^{\Pi} = \bs + \bt$. Thus, it suffices to give DP algorithms for approximately computing $\bs, \bt$.
For computing the estimate of $\bs$, we can simply use Laplace and vanilla Gaussian mechanism, because the sensitivity is now reduced as we only compare candidates in the same bucket.

For computing the estimate of $\bt$, we encode $t_{uv}$ as a two-way marginal query as follows. First, let $d = m \cdot B$ where we associate $[d]$ naturally with $[m] \times [B]$. Then, for each $i \in [n]$, we let $x_i \in \{0, 1\}^d$ be defined by $x_i(u,b)=\bone[\iota(\pi_i(u))=b]$ for all $u \in [m], b \in [B]$. The crucial observation here is that $\bt$ can now be written in terms of marginal queries: $$t_{uv} = \sum_{b_u,b_v\in [B] \atop b_u<b_v} q_{\{(u,b_u),(v,b_v)\}}(\bX),$$
As such, we can use the aforementioned DP two-way marginal query algorithms to estimate $\bt$. 
Our algorithm is shown in \Cref{alg:ptas_small}. The choice of distribution $\cD$ and subroutine $\maralg$ depends on whether we aim for $\eps$-DP or $(\eps,\delta)$-DP. The operation $\clip$, which ensures that the outputs lie within the interval $[0, 1]$\ifaaaicameraready \else, is formally defined in the appendix\fi. It is simple to analyze the error of \Cref{alg:ptas_small}. By optimizing the parameter $B$ to minimize the error and invoking \Cref{lem:err-to-apx}, we arrive at our PTAS for small $n$. 

\begin{algorithm}
    \caption{DP Approximation of $\bw$ for Small $n$}
    \label{alg:ptas_small}
\textbf{Input:} Dataset $\Pi = (\pi_1, \dots, \pi_n) \in (\bbS_m)^n$\\
\textbf{Output:} Estimate $\tbw^{\Pi}$ of $\bw^{\Pi}$ \\
\textbf{Parameters: } Distribution $\cD$, \# of buckets $B$, DP two-way marginal algo $\maralg$  
\begin{algorithmic}[1]
\FOR[Noising $s_{uv}$]{$u, v \in [m]$}
\STATE $\ts_{uv} \gets s_{uv} + r_{uv}$ where $r_{uv} \sim \cD$
\ENDFOR
\FOR{$i \in [n]$}
\STATE $x_i \in \{0, 1\}^{mB}$ be s.t. $x_i(u,b)=\bone[\iota(\pi_i(u))=b]$
\ENDFOR
\STATE $(\tq_S)_{S \in \binom{[d]}{2}} \gets$ output of $\maralg(x_1, \dots, x_n)$
\FOR[Estimate $t_{uv}$ from marginals]{$u, v \in [m]$}
\STATE $\ttt_{uv} = \sum_{b_u,b_v\in [B] \atop b_u<b_v} \tq_{\{(u,b_u),(v,b_v)\}}$
\ENDFOR
\RETURN $\clip(\ts_{uv} + \ttt_{uv})_{u,v \in [m]}$
\end{algorithmic}
\end{algorithm}

\subsection{Large $n$: Leveraging Unbiasedness}

As mentioned in \Cref{subsec:overview}, if we were able to use the noised weights without any clipping, then we would have been done. Unfortunately, this is not possible since the weights could become negative due to the noises. It turns out that, when $n$ is large, we can handle this as follows. 

Consider each pair $u, v \in [m]$. First, consider the ``balanced'' case where $w_{uv}$ and $w_{vu}$ are both not too small, say $w_{vu}, w_{vu} \in [\frac{1}{6}, \frac{5}{6}]$. In this case, when $n$ is sufficiently large, the noise added has such a small variance that with high probability all the values $w_{uv}, w_{vu}$ remain non-negative after noise addition. Hence, we can keep this case as is.

Next, consider the ``imbalanced'' case where either $w_{uv}$ or $w_{vu}$ is smaller than $1/6$. Assume wlog that $w_{vu} < \frac{1}{6}$. In this case, we replace $w_{vu}$ with zero and $w_{uv}$ with $w_{uv} - w_{vu}$. While this new instance is \emph{not} a Kemeny rank aggregation instance anymore, it still is a bounded WFAS instance, meaning that we can apply the PTAS from \Cref{thm:MS}. We can show that this also yields a $(1 + \xi)$-approximation for the original Kemeny rank aggregation instance.

To add noise in the imbalanced case, notice that $w_{uv} - w_{vu} \geq \frac{2}{3}$. Thus, we can add noise to $w_{uv} - w_{vu}$ and leave the other term zero. One can argue that, with high probability, all the values $w_{uv} - w_{vu}$ 
remains non-negative after noise addition. In the balanced case, we simply add noise to both $w_{uv}, w_{vu}$ as usual. As already mentioned above, since they are both at least $\frac{1}{6}$ beforehand, they remain non-negative after noise addition with high probability.

In our full algorithm--presented in \Cref{alg:ptas_large}, we need one additional step in order to determine which case each pair $u,v$ belongs to. This is because directly checking if e.g. $w_{uv} \in [\frac{1}{6}, \frac{5}{6}]$ violates DP. Nevertheless, this step turns out to be simple: We can add noises to all of $w_{uv}$ and classify them accordingly. As with the small $n$ case, the noise distribution $\cD$ here is either Laplace or Gaussian based on whether we desire $\eps$-DP or $(\eps,\delta)$-DP.

The crux of the proof is to show that, with high probability, the following two events hold: (i) the instance $\tbw'$ is bounded, and (ii) for \emph{all} permutations $\pi$, the cost difference $\cost_{\tbw}(\pi) - \cost_{\tbw'}(\pi)$ has small magnitude. (i) follows from standard concentration of the noise. As for (ii), we can write the cost difference as a \emph{summation} of at most $m^2$ (independent) noises, and we argue that this sum is highly concentrated. For Gaussian noise, this is again simple since the sum of noises is also a Gaussian. However, for Laplace noise, we need to use a result of \cite{ChanSS11} on the concentration of a sum of Laplace random variables.

\begin{algorithm}[h!]
    \caption{DP PTAS for Large $n$}
    \label{alg:ptas_large}
\textbf{Input:} Dataset $\Pi = (\pi_1, \dots, \pi_n) \in (\bbS_m)^n$\\
\textbf{Output:} Aggregated rank $\psi \in \bbS_m$ \\
\textbf{Parameters: } Noise distribution $\cD$
\begin{algorithmic}[1]
\STATE $\Zl \gets \emptyset, \Zs \gets \emptyset$
\FORALL[Pairwise Classification]{$1 \leq u < v \leq m$}
\STATE $w'_{uv} \gets w^{\Pi}_{uv} + \theta_{uv}$ where $\theta_{uv}\sim D$
\IF{$w'_{uv} > 5/6$}
\STATE Add $(u,v)$ to $\Zl$
\ELSIF{$w'_{uv} < 1/6$}
\STATE Add $(v,u)$ to $\Zl$
\ELSE[$\frac{1}{6} \le w'_{uv} \le \frac{5}{6}$]
\STATE Add $(u,v)$ to $\Zs$
\ENDIF
\ENDFOR
\FORALL[Noise Addition: Imbalanced]{$(u, v) \in \Zl$}
\STATE $\tw'_{uv} \gets w_{uv} - w_{vu} + r_{uv}$ where $r_{uv} \sim \cD$
\STATE $\tw'_{vu} \gets 0$
\ENDFOR
\FORALL[Noise Addition: Balanced]{$(u, v) \in \Zs$}
\STATE $\tw'_{uv} \gets w_{uv} + r_{uv}$ where $r_{uv} \sim \cD$
\STATE $\tw'_{vu} \gets 1 - \tw'_{uv}$
\ENDFOR
\IF[Approximation]{$\tbw'$ is a bounded instance}
\RETURN Output from running the algorithm from \Cref{thm:MS} on $\tbw'$
\ELSE
\RETURN random permutation from $\bbS_m$
\ENDIF
\end{algorithmic}
\end{algorithm}

\section{Conclusion and Discussion}

We study rank aggregation problems under DP, both in the central and local models. For footrule rank aggregation, we give a polynomial-time algorithm with nearly optimal errors in the central model, which translates to 2-approximation algorithms for Kemeny rank aggregation. We also improve the additive error in the PTAS for the problem compared to \cite{AlabiGKM22}. The obvious open question here is to close the gap in terms of the additive errors of the PTAS. In particular, for $\eps$-DP and $(\eps, \delta)$-DP, our upper bounds are $\frac{1}{\eps n} \cdot \tO(m^{27/7})$ and $\frac{1}{\eps n} \cdot \tO(m^{65/22})$, respectively, whereas the lower bounds from \cite{AlabiGKM22} are only $\frac{1}{\eps n} \cdot \Omega(m^{3})$ and $\frac{1}{\eps n} \cdot \Omega(m^{2.5})$, respectively. Another interesting research direction is to prove lower bounds for LDP; to the best of our knowledge, no previous work has pursued this direction.

\subsubsection{Acknowledgments.}
We are grateful to Jittat Fakcharoenphol for hosting us and for helpful discussions during the early stages of the project. The project was done while PV was working with JF at Kasetsart University, Thailand. PV is supported by Graduate Student Scholarship, Faculty of Engineering, Kasetsart University, Grant Number 65/06/COM/M.Eng and NCN grant number 2020/39/B/ST6/01641. VS is supported by JST
NEXUS Grant Number Y2024L0906031 and KAKENHI Grant JP25K00369. At the time of this project, QH was affiliated with the University of Tokyo and supported by KAKENHI Grant 20H05965, and by JST SPRING Grant Number JPMJSP2108. QH is currently part of BARC supported by the VILLUM Foundation grant 54451.

\bibliography{ref}

\ifaaaicameraready \else
\appendix

\section{Comparison with Related Works}

\label{app:comparison}

In addition to the exponential mechanism discussed in Section~\ref{sec:contribution}, \citet{HayEM17} also proposes two other mechanisms. However, these mechanisms do not come with theoretical guarantees on the parameters $\alpha$ or $\beta$.

It is discussed in~\cite{AlabiGKM22} that the expected cost of these two mechanisms includes a term that is independent of the number of rankings $n$. As a result, when $n$ is significantly larger than the number of candidates $m$, their expected cost can be substantially higher than that of our proposed methods and those in \cite{AlabiGKM22}.

\section{Additional Preliminaries and Tools}
\label{app:add-prelim}

We sometimes write $f \lesssim g$ as a shorthand for $f = O(g)$.

\subsection{Differential Privacy: Basics}

\subsubsection{Concentrated DP.}
We will use a variant of DP, called zero-concentrated differential privacy (zCDP), defined below, since it is more convenient for composition theorems and the Gaussian mechanism compared to approx-DP.
\begin{definition}[$\rho$-zCDP~\cite{BunS16}]
  Let $\rho > 0$. A randomized ranking aggregation $\mathcal{M}: \bbS_m \times \cX^n$ is $\rho$-zCDP if, for any $\alpha \in (1, \infty)$, neighboring datasets $\Pi, \Pi' \in \cX^n$, 
$$D_\alpha\left( \mathcal{M}(D) \,\|\, \mathcal{M}(D') \right) \leq \rho \cdot \alpha,$$
where \( D_\alpha(P \| Q) \) denotes the Rényi divergence of order \( \alpha \) between distributions \( P \) and \( Q \), defined as
$$D_\alpha(P \| Q) = \frac{1}{\alpha - 1} \log \mathbb{E}_{x \sim Q} \left[ \left( \frac{P(x)}{Q(x)} \right)^\alpha \right].$$
\end{definition}

Note that it is simple to transform $\rho$-zCDP guarantee to an approx-DP guarantee:
\begin{theorem}[\citealt{BunS16}] \label{thm:cdp-to-apx-dp}
Any $\rho$-zCDP mechanism satisfies $(\rho + \sqrt{\rho \log(1/\delta)}, \delta)$-DP.
\end{theorem}

\subsubsection{Basic Mechanisms.}
Recall that, for $p \geq 1$, the $\ell_p$-norm of a vector $u \in \R^d$ is defined as $\left(\sum_{i \in [d]} u_i^p\right)^{1/p}$.

\begin{definition}[Sensitivity]
For a function $f: \cX^n \to \R^d$, its $\ell_p$-sensitivity, denoted by $\Delta_p(f)$, is defined as $\max_{\bX, \bX'} \|f(X) - f(X')\|_p$ where the maximum is over all pairs of neighboring datasets $\bX, \bX'$.   
\end{definition}

We next recall two basic mechanisms in DP: The Laplace and Gaussian mechanisms.
The Laplace distribution with parameter $b$, denoted by $\Lap(b)$, is supported on $\R^d$ where its density at $x$ is proportional to $\exp(-\|x\|_1 / b)$. 

\begin{theorem}[Laplace Mechanism~\cite{DworkMNS06}] \label{thm:laplace}
The Laplace mechanism for a function $f$ outputs $f(\bX) + Z$ where $Z \sim \Lap(\Delta_1(f) / \eps)$. It satisfies $\eps$-DP.
\end{theorem}

The (spherical) Gaussian distribution with standard deviation $\sigma$, denoted by $\cN(0, \sigma^2 I_d)$, is supported on $\R^d$ where its density at $x$ is proportional to $\exp(-0.5 \|x\|_2^2 / \sigma^2)$.

\begin{theorem}[Gaussian Mechanism~\cite{BunS16}] \label{thm:gaussian}
The Gaussian mechanism for a function $f$ outputs $f(\bX) + Z$ where $Z \sim \cN(0, \sigma^2)$ where $\sigma = \Delta_2(f) / \sqrt{2\rho}$. It satisfies $\rho$-zCDP.
\end{theorem}

\subsubsection{Composition Theorems.}
We state the composition theorems for pure-DP and zCDP in the following two theorems:
\begin{theorem}[Composition for Pure DP] \label{thm:comp-pureDP}
Let $\mathcal{M}_1, \ldots, \mathcal{M}_k$ be randomized algorithms where each $\mathcal{M}_i$ satisfies $\varepsilon_i$-DP. Then, the sequence of mechanisms $(\mathcal{M}_1, \ldots, \mathcal{M}_k)$ satisfies $\left(\sum_{i = 1}^k \varepsilon_i\right)$-differential privacy.
\end{theorem}

\begin{theorem}[Composition for zCDP] \label{thm:comp-CDP}
Let $\mathcal{M}_1, \ldots, \mathcal{M}_k$ be mechanisms such that each $\mathcal{M}_i$ satisfies $\rho_i$-zCDP. Then the composition $(\mathcal{M}_1, \ldots, \mathcal{M}_k)$ satisfies $\left( \sum_{i=1}^k \rho_i\right)$-zCDP.
\end{theorem}

\subsection{DP Algorithms for Two-way Marginals}
Recall the two-way marginal queries defined in \Cref{def:2-marginal}. Below we discuss known $\eps$-DP and $\rho$-CDP algorithms for releasing two-way marginals, which will be used in our algorithm (\Cref{alg:ptas_small}).

\subsubsection{Pure-DP.} For $\eps$-DP, we use the algorithm of \cite{Nikolov24}, whose guarantee can be stated as follows.

\begin{theorem}
\label{cor:pure-2way}
For any $d\ge 1$, there exists a polynomial-time $\varepsilon$-DP algorithm given $\bX$, outputs $(\tq_S)_{S \in \binom{[d]}{2}}$ such that
\[
    \E\left[\sum_{S \in \binom{[d]}{2}} |\tq_S - q_S(\bX)|\right]\le O\left(\sqrt{\frac{d^5}{n\varepsilon}}\right).
\]
\end{theorem}

Since \Cref{cor:pure-2way} is not stated in the same way as in \cite{Nikolov24}, we now explain how to derive the theorem of above form from that paper. To do so, we need to recall a few additional terminologies.

For a dataset $\bX\in\mathcal{X}^n$ and query workloads $Q=\{q_1,...,q_{\binom{d}{2}}\}$ of all-pairs 2-way marginal query, we write $Q(\bX)$ for the vector $(q_1(\bX),...,q_{\binom{d}{2}}(\bX))$.

Define the ($\ell_2$-)error of a mechanism $\mathcal{M}$ on a workload $Q$ and datasets of size $n$ as \[
    \mathrm{err}(\mathcal{M},Q,n)=\max_{\bX\in\mathcal{X}^n} \mathbb{E}\left[\frac{1}{\sqrt{\binom{d}{2}}}\|\mathcal{M}(\bX)-Q(\bX)\|_2\right]
\]
and the sample complexity of the mechanism on queries $Q$  as \[
    \mathrm{sc}(\mathcal{M},Q,\alpha)=\inf\{n: \mathrm{err}(\mathcal{M},Q,n)\le \alpha\}.
\]

\begin{theorem}[Theorem 23 in \cite{Nikolov24}] \label{thm:nikolov-sc}
There exists a polynomial-time $\varepsilon$-DP algorithm $\mathcal{M}$ whose sample complexity for the workload of 2-way marginals is \[
    \mathrm{sc}(\mathcal{M},Q,\alpha) \lesssim \min\left\{\frac{d^{1.5}}{\alpha \varepsilon},\frac{d}{\alpha^2 \varepsilon}\right\}.
\]
\end{theorem}

We can now prove \Cref{cor:pure-2way}.

\begin{proof}[Proof of \Cref{cor:pure-2way}]
From \Cref{thm:nikolov-sc}, we have $\mathrm{sc}\left(\mathcal{M},Q,\Theta\left(\sqrt{\frac{d}{n\varepsilon}}\right)\right)\leq n$. Thus, we can conclude that \[
    \mathrm{err}(\mathcal{M},Q,n) \le O\left(\sqrt{\frac{d}{n\varepsilon}}\right).
\]
Note that
\begin{align*}
\sum_{S \in \binom{[d]}{2}} |\tq_S - q_S(\bX)| &= \|\mathcal{M}(\bX)-Q(\bX)\|_1 \\
&\leq \sqrt{\binom{d}{2}}  \|\mathcal{M}(\bX)-Q(\bX)\|_2.
\end{align*}
Thus, the expectation of the LHS is at most $\binom{d}{2}$ times $\mathrm{err}(\mathcal{M},Q,n)$.
\end{proof}

\subsubsection{zCDP.}
For zCDP, we use an algorithm of \cite{DworkNT15} which has a more specific error guarantee. In particular, we can specify a distribution $\cP$ on the queries, and the error is only measured with respect to this distribution, as stated more precisely below. This will be helpful for our algorithm (\Cref{alg:ptas_small}) since we will only be focusing on subsets of two-way marginal queries.

\begin{theorem}[\citealt{DworkNT15}] \label{thm:2-way-marginal}
For any distribution $\cP$ on $\binom{[d]}{2}$, there exists a polynomial-time $\rho$-zCDP\footnote{Although \cite{DworkNT15} only focuses on $(\eps, \delta)$-DP, they only use the Gaussian mechanism. Thus, their algorithm naturally satisfies zCDP too.} algorithm $\cA$ that, given $\bX$, outputs $(\tq_S)_{S \in \binom{[d]}{2}}$ such that
\begin{align*}
\E_{\cA} \E_{S \sim \cP}[|\tq_S - q_S(\bX)|] \leq O\left(\frac{d^{1/4}}{\sqrt{n} \cdot \rho^{1/4}}\right).
\end{align*}
\end{theorem}

\subsection{From High-Probability to Expected Guarantees}

Sometimes, it will be more convenient to analyze a high-probability additive error guarantee of an algorithm rather than an expectation. Below, we give a lemma that translates such a high-probability guarantee to an expectation, provided that we have another algorithm with the same approximation ratio but possibly worse additive errors.

For simplicity, we will only focus on Kemeny rank aggregation here. Let $\opt$ denote the optimum, i.e. $\opt := \min_{\pi \in \bbS_m} \cost_{\bw}(\pi)$, and let $\pi^*$ denote an optimal ranking, i.e. $\pi^* = \argmin_{\pi \in \bbS_m} \cost_{\bw}(\pi)$.
We say that a solution $\pi$ is an $(\alpha, \beta)$-approximation if $\cost_{\bw}(\pi) \leq \alpha \cdot \opt + \beta$. 

\begin{lemma} \label{lem:high-prob-to-exp}
Let $\alpha, \beta, \beta' > 0$ and $\zeta \in (0, 1/2)$.

Suppose that there exist $\frac{\eps}{3}$-DP efficient algorithms $\cM_1, \cM_2$ algorithms such that
\begin{itemize}
\item $\cM_1$ outputs an $(\alpha, \beta)$-approximate solution with probability $1 - \zeta$.
\item $\cM_2$ is an $(\alpha, \beta')$-approximation algorithm.
\end{itemize}
Then, there is an efficient $\eps$-DP algorithm $\cM$, which is an $\left(\alpha, O\left(\beta + \zeta \cdot \beta' + \frac{m^2}{\eps n}\right)\right)$-approximation algorithm.
\end{lemma}

Similarly, suppose that there exist $\frac{\rho}{3}$-zCDP efficient algorithms $\cM_1, \cM_2$ algorithms such that
\begin{itemize}
\item $\cM_1$ outputs an $(\alpha, \beta)$-approximate solution with probability $1 - \zeta$.
\item $\cM_2$ is an $(\alpha, \beta')$-approximation algorithm.
\end{itemize}
Then, there is an efficient $\rho$-DP algorithm $\cM$, which is an $\left(\alpha, O\left(\beta + \zeta \cdot \beta' + \frac{m^2}{\sqrt{\rho} \cdot n}\right)\right)$-approximation algorithm.
\begin{proof}
Let us focus on the pure-DP case. The algorithm $\cM$ on input $\Pi$ works as follows:
\begin{itemize}
\item Run $\cM_1$ on $\Pi$ to get an output $\pi_1$
\item Run $\cM_2$ on $\Pi$ to get an output $\pi_2$
\item Compute $\cost_{\bw^{\Pi}}(\pi_1) - \cost_{\bw^{\Pi}}(\pi_2) + r$ where $r \sim \Lap\left(\frac{6m^2}{\eps n}\right)$. If this is positive, then output $\pi_2$. Otherwise, output $\pi_1$.
\end{itemize}
To see that this satisfies $\eps$-DP, note that each of the first two steps satisfies $\frac{\eps}{3}$-DP due to our assumption on $\cM_1, \cM_2$. As for the last step, notice that $\cost_{\bw^{\Pi}}(\pi_1) - \cost_{\bw^{\Pi}}(\pi_2)$ has sensitivity at most $\frac{2m^2}{\eps n}$. Thus, \Cref{thm:laplace} ensures that this satisfies $\frac{\eps}{3}$-DP. Finally, applying the composition theorem (\Cref{thm:comp-pureDP}) ensures that $\cM$ is $\eps$-DP.

Next, we analyze the approximation guarantee of $\cM$. Let $\pi$ be the output of $\cM$. 
Due to the last step, notice that we have
\begin{align*}
\cost_{\bw^{\Pi}}(\pi) \leq |r| + \min\{\cost_{\bw^{\Pi}}(\pi_1), \cost_{\bw^{\Pi}}(\pi_2)\}.
\end{align*}
Taking expectation on both side, we have
\begin{align*}
&\E[\cost_{\bw^{\Pi}}(\pi)] \\
&\leq O\left(\frac{m^2}{\eps n}\right) + \E[\min\{\cost_{\bw^{\Pi}}(\pi_1), \cost_{\bw^{\Pi}}(\pi_2)\}]
\end{align*}

Let $\cE$ denote the event that $\pi_1$ is an $(\alpha, \beta)$-approximation solution. By the assumption on $\cM_1$, we have $\Pr[\neg \cE] \leq \zeta$. Thus, we can bound the expectation of the last term above as follows:
\begin{align*}
&\E[\min\{\cost_{\bw^{\Pi}}(\pi_1), \cost_{\bw^{\Pi}}(\pi_2)\}] \\
&= \Pr[\cE] \cdot \E[\min\{\cost_{\bw^{\Pi}}(\pi_1), \cost_{\bw^{\Pi}}(\pi_2)\}~\mid~\cE] + \\
&\qquad \Pr[\neg \cE] \cdot \E[\min\{\cost_{\bw^{\Pi}}(\pi_1), \cost_{\bw^{\Pi}}(\pi_2)\}~\mid~\neg \cE] \\
&\leq \Pr[\cE] \cdot \E[\cost_{\bw^{\Pi}}(\pi_1)~\mid~\cE] + \\
&\qquad \Pr[\neg \cE] \cdot \E[\cost_{\bw^{\Pi}}(\pi_2)~\mid~\neg \cE] \\
&\leq \Pr[\cE] \cdot (\alpha \cdot \opt + \beta) + \Pr[\neg \cE] \cdot (\alpha \cdot \opt + \beta') \\
&\leq \alpha \cdot \opt + \beta + \zeta \cdot \beta'.
\end{align*}
Combining the above inequalities, we can conclude that $\cM$ is an $\left(\alpha, O\left(\beta + \zeta \cdot \beta' + \frac{m^2}{\eps n}\right)\right)$-approximation algorithm as desired.

The proof for zCDP case is exactly identical except that we use Gaussian noise instead.
\end{proof}

\subsection{Clipping}

Below, we define the $\clip$ operator used in \Cref{alg:ptas_small}. First, $\clip: \R^{m \times m} \to [0,1]^{m \times m}$ such that the output always satisfies boundness (\Cref{def:bounded}). Let $\cW = \{\bw \in [0, 1]^n \mid \forall u, v \in [m], w_{uv} + w_{vu} = 1\}$. The operation is defined as follows:
\begin{align*}
\clip(\tbw) = \argmin_{\tbw' \in \cW} \|\tbw' - \tbw\|_1.
\end{align*}
In other words, $\clip$ projects $\tbw$ back onto $\cW$ under the $\ell_1$ distance. We remark that $\clip$ is efficient since $\cW$ is convex and the $\ell_1$ norm is also convex. Moreover, $\clip$ can be computed quickly since we only have to consider each pair $(\tw_{uv}, \tw_{vu})$ when computing $\tw'_{uv}, \tw'_{vu}$.

The following is a simple but useful observation:
\begin{observation} \label{obs:clipping}
For any $\bw \in \cW$ and $\tbw \in \R^d$, we have $\|\clip(\tbw) - \bw\|_1 \leq 2 \cdot \|\tbw - \bw\|_1$
\end{observation}
\begin{proof}
By triangle inequality and definition of $\clip$, we have
\begin{align*}
\|\clip(\tbw) - \bw\|_1 &\leq \|\clip(\tbw) - \tbw\|_1 + \|\tbw - \bw\|_1 \\
&\leq \|\bw - \tbw\|_1 + \|\tbw - \bw\|_1 \\
&= 2 \cdot \|\tbw - \bw\|_1. \qedhere
\end{align*}
\end{proof}

\section{Lower Bounds}
\label{app:lb}

In this section, we discuss the lower bounds for the rank aggregation problems.

\subsection{Kemeny Rank Aggregation}

We start by recalling the lower bounds proved in \cite{AlabiGKM22}:

\begin{theorem}[\citealt{AlabiGKM22}]
\label{thm:lower-bound1}
For any $\alpha, \eps > 0$, there exists $\eta$ (depending on $\alpha$) such that no $\eps$-DP $(\alpha, 0.01m^2)$-approximation algorithm for Kemeny rank aggregation for $n \leq \eta \cdot m/\eps$.
\end{theorem}

\begin{theorem}[\citealt{AlabiGKM22}]
\label{thm:lower-bound2}
For any $\alpha > 0, \eps \in (0, 1]$, there exist $\eta, c > 0$ (depending on $\alpha$) such that there is no $(\eps, o(1/n))$-DP $(\alpha, 0.01m^2)$-approximation algorithm for Kemeny rank aggregation for $n \leq \eta \cdot \sqrt{m}/\eps$.
\end{theorem}

Suppose there exists an $\varepsilon$-DP approximation algorithm with additive error $o\left(\frac{m^3}{\varepsilon n}\right)$. By applying this algorithm to a dataset of size $n = \eta \cdot m / \eps$, we would obtain an additive error of $o\left(m^2\right)$, which contradicts \Cref{cor:tight1}. Similarly, by \Cref{cor:tight2}, there cannot exist an $(\varepsilon,\delta)$-DP algorithm with additive error $o\left(\frac{m^{2.5}}{\varepsilon n} \right)$.

In other words, the additive error of $\tilde{O}\left(\frac{m^3}{\eps n}\right)$ and $\tilde{O}\left(\frac{m^{2.5}}{\eps n}\right)$ achieved in \Cref{thm:two-apx} for $\eps$-DP and $(\eps, \delta)$-DP, respectively, are nearly optimal.

\subsection{Footrule Rank Aggregation}

Next, we will argue that the above lower bounds translate to similar lower bounds for footrule rank aggregation as well.
%
To do this, we will use the following well-known relationship between the two distances.
\begin{lemma}[\citealt{DG77}] \label{lem:dist-relation}
For any $\psi, \pi \in \bbS_m$, we have $K(\psi, \pi) \leq F(\psi, \pi) \leq 2K(\psi, \pi)$.
\end{lemma}

Using the above, we can derive the following generic result relating approximation algorithms for the two problems:

\begin{lemma} \label{lem:reduction-generic}
Any algorithm that is an $(\alpha, \beta)$-approximation algorithm for footrule rank aggregation is also an $(2\alpha, \beta)$-approximation algorithm for Kemeny rank aggregation.
\end{lemma}

\begin{proof}
Let $\psi$ be the output of the algorithm. Furthermore, let $\psi^* = \argmin_{\psi'} K(\psi', \Pi)$. From \Cref{lem:dist-relation}, we have that 
\begin{align*}
K(\psi, \Pi) - 2\alpha \cdot K(\psi^*, \Pi) \leq F(\psi, \Pi) - \alpha \cdot F(\psi^*, \Pi).
\end{align*}
Since the algorithm is an $(\alpha, \beta)$-approximation algorithm for footrule rank aggregation, the expectation of the RHS is at most $\beta$. Thus, we also have that the algorithm is $(2\alpha, \beta)$-approximation for Kemeny rank aggregation.
\end{proof}

Thus, plugging \Cref{lem:reduction-generic} into \Cref{thm:lower-bound1,thm:lower-bound2}, we arrive at the following corollaries.

\begin{corollary}
For any $\alpha, \eps > 0$, there exists $\eta$ (depending on $\alpha$) such that there is  no $\eps$-DP $(\alpha, 0.01m^2)$-approximation algorithm for footrule rank aggregation for $n \leq \eta \cdot m/\eps$. \label{cor:tight1}
\end{corollary}

\begin{corollary}
For any $\alpha > 0, \eps \in (0, 1]$, there exists $\eta, c > 0$ (depending on $\alpha$) such that there is no $(\eps, o(1/n))$-DP $(\alpha, 0.01m^2)$-approximation algorithm for footrule rank aggregation for $n \leq \eta \cdot \sqrt{m}/\eps$. \label{cor:tight2}
\end{corollary}

Using similar arguments as before, we can conclude that the additive errors $\tilde{O}\left(\frac{m^3}{\varepsilon n} \right)$ and $\tilde{O}\left(\frac{m^{2.5}}{\varepsilon n} \right)$ achieved in \Cref{thm:red-median-to-rank-agg} for $\varepsilon$-DP and $(\varepsilon,\delta)$-DP, respectively, are nearly optimal.

\section{Missing Proofs from \Cref{sec:median}}
\label{app:median}

To formalize the guarantees of \Cref{algo:modified-binary}, we need to formalize vector aggregation algorithms and their guarantees.

In the \emph{$\ell_p$-norm vector aggregation} problem, the domain $\cX$ is the set of all vectors in $\R^d$ such that its $\ell_p$-norm is at most $C$ (where both $d$ and $C$ are parameters). On input dataset $\bW = (w_1, \dots, w_n) \in \cX^n$, the goal is to output an estimate $\tw$ of $\ow := \frac{1}{n} \sum_{i\in[n]} w_i$. Again, we say that an algorithm is $\beta$-accurate if its output $\tw$ satisfies $\E[\|\tw - \ow\|_\infty] \leq \beta$.

The Laplace and Gaussian mechanisms (\Cref{thm:laplace,thm:gaussian}) can be stated in terms of $\ell_1$-norm and $\ell_2$-norm vector aggregation as follows. (Note that these follows directly from concentration of Laplace and Gaussian.)

\begin{theorem} \label{thm:lap-vec-agg}
There is an efficient $\eps$-DP algorithm for $\ell_1$-norm vector aggregation that is $O\left(\frac{C \log d}{\eps n}\right)$-accurate.
\end{theorem}

\begin{theorem} \label{thm:gau-vec-agg}
There is an efficient $(\eps, \delta)$-DP algorithm for $\ell_2$-norm vector aggregation that is $O\left(\frac{C \sqrt{\log d \log(1/\delta)}}{\eps n}\right)$-accurate.
\end{theorem}

Finally, for $\eps$-LDP, we use the algorithm of \cite{duchi2014localprivacydataprocessing}, whose guarantee is stated below. (Note that this guarantee follows from \cite{BlasiokBNS19} who shows that the noise is sub-Gaussian.)

\begin{theorem}[\citealt{duchi2014localprivacydataprocessing,BlasiokBNS19}] \label{thm:ldp-vec-agg}
There is an efficient $\eps$-LDP algorithm for $\ell_2$-norm vector aggregation that is $O\left(\frac{C \sqrt{\log d}}{\eps \sqrt{n}}\right)$-accurate.
\end{theorem}

\subsection{Proof of \Cref{thm:apx-median}}

For the three settings, we instantiate \Cref{algo:modified-binary} with \Cref{thm:lap-vec-agg,thm:gau-vec-agg,thm:ldp-vec-agg}, respectively.

\paragraph{Privacy.} The DP guarantees immediately follows since we use the $\eps$-DP, $(\eps, \delta)$-DP and $\eps$-LDP algorithms for vector aggregation, and only use post-processing afterwards. 

\paragraph{Accuracy.}
Throughout this proof, we assume that $\kappa$ is a constant such that $1 < \kappa < 2$. A first key lemma here is that the error of the final estimates is at most that of the aggregated vector, as stated and proved formally below.

\begin{lemma} \label{lem:accuracy-transfer-bin-tree}
Let $v^{\agg, \wei} = \frac{1}{n} \sum_{i \in [n]} v^{\agg, \wei}_i$ and $w^{\agg, \wei} = \frac{1}{n} \sum_{i \in [n]} w^{\agg, \wei}_i$ be the weighted aggregates without any noise. Then, we have
$\|\tgamma - \gamma(\bx)\|_\infty \leq O(1) \cdot \max\{\|v^{\agg, \wei} - \tv^{\agg, \wei}\|_\infty, \|w^{\agg, \wei} - \tw^{\agg, \wei}\|_\infty\}$.
\end{lemma}
\begin{proof}
Let $\beta = \max\{\|v^{\agg, \wei} - \tv^{\agg, \wei}\|_\infty, \|w^{\agg, \wei} - \tw^{\agg, \wei}\|_\infty\}$. Let us fix $j \in [m]$. We first observe that, if we replace $\tv^{\agg, \wei}, \tu^{\agg, \wei}$ with $v^{\agg, \wei}, u^{\agg, \wei}$, we would have $\tgamma_j = \gamma_j(\bx)$ for all $j$. Thus, we can bound $|\tgamma_j - \gamma_j(\bx)|$ as follows:
\begin{align*}
&\left|\tgamma_j - \gamma_j(\bx)\right| \\
&= \Big|\sum_{t \in \cB_j} s_{j, t} \cdot \kappa^{\ell(t)-d} \cdot \Big(\left(\tv^{\agg, \wei}_{t'} - v^{\agg, \wei}_{t'}\right) \\ & \qquad\qquad\qquad + \frac{r(t') - j}{2^{\ell(t)}} \cdot \left(\tu^{\agg, \wei}_{t'} - u^{\agg, \wei}_{t'}\right)\Big)\Big| \\
&\leq \sum_{t \in \cB_j} \kappa^{\ell(t)-d} \cdot \left(\beta +  2\beta\right) \\
&\leq 3\beta \cdot \sum_{\ell=0}^d \kappa^{\ell-d} \cdot  \\
&\leq O(\beta),
\end{align*}
where the first inequality follows from triangle inequality and from $|s_{j, t}| = 1$, and $\left|\frac{r(t') - j}{2^{\ell(t)}}\right| \leq 2$.

Hence, we can conclude that $\|\tgamma - \gamma(\bx)\|_\infty \leq O(\beta)$.
\end{proof}

A second key lemma bounds the sensitivity of the weighted contribution vectors:
\begin{lemma} \label{lem:norm-bound}
For all $i \in [n]$ and $p \in \{1, 2\}$, we have $\|v^{\wei}_i\|_p, \|u^{\wei}_i\|_p \leq O(m)$.
\end{lemma}
\begin{proof}
We have
\begin{align*}
\|v^{\wei}_i\|_p^p &= \sum_{t \in \cB \atop I(t) \ni x_i} \left(\kappa^{d - \ell(t)}  \cdot (x_i - r(t))\right)^p \\
&\leq \sum_{t \in \cB \atop I(t) \ni x_i} \left(\kappa^{d - \ell(t)}  \cdot 2^{\ell(t)}\right)^p \\
&= m^p \cdot \sum_{t \in \cB \atop I(t) \ni x_i} (\kappa/2)^{p(d - \ell(t))} \\
&= m^p \cdot \sum_{\ell=0}^d (\kappa/2)^{p(d - \ell)} \\
&= \Theta(m^p),
\end{align*}
where the penultimate equality follows from the fact that $x_i$ is contained in only one interval per level.

Similarly, we have
\begin{align*}
\|u^{\wei}_i\|_p^p &= \sum_{t \in \cB \atop I(t) \ni x_i} \left(\kappa^{d - \ell(t)}  \cdot 2^{\ell(t)}\right)^p \\
&= m^p \cdot \sum_{t \in \cB \atop I(t) \ni x_i} (\kappa/2)^{p(d - \ell(t))} \\
&= m^p \cdot \sum_{\ell=0}^d (\kappa/2)^{p(d - \ell)} \\
&= \Theta(m^p). \qedhere
\end{align*}
\end{proof}

We can now easily conclude the proof of \Cref{thm:apx-median}. \Cref{lem:norm-bound} (together with the guarantees in \Cref{thm:lap-vec-agg,thm:gau-vec-agg,thm:ldp-vec-agg}) implies the the vector aggregation used in \Cref{algo:modified-binary} is $\beta$-accurate for the following $\beta$ (note that $d = 2|\cB| = O(m)$ here):
\begin{itemize}
\item $\beta = O\left(\frac{m \log m}{\eps n}\right)$ for $\eps$-DP,
\item $\beta = O\left(\frac{m \sqrt{\log m \log(1/\delta)}}{\eps n}\right)$ for $(\eps, \delta)$-DP,
\item $\beta = O\left(\frac{m \sqrt{\log m}}{\eps \sqrt{n}}\right)$ for $\eps$-LDP.
\end{itemize}
\Cref{lem:accuracy-transfer-bin-tree} then implies that the same accuracy guarantees hold for the final estimates $(\tgamma_j)_{j \in [m]}$.

\subsection{Proof of \Cref{cor:par-apx-median}}

This simply follows by running the algorithm from \Cref{thm:apx-median} in parallel (i.e. one to compute approximate median for $x_{1,q}, \dots, x_{m,q}$ for each $q \in [m]$). Notice that we can still use the same vector aggregation algorithm on the concatenated vector, except that the $\ell_1$-norm of each vector increases by a factor of $m$ whereas the $\ell_2$-norm of each vector increases by a factor of $\sqrt{m}$. Thus, the error in $\eps$-DP case increases by a factor of $m$, whereas the errors in $(\eps,\delta)$-DP and $\eps$-LDP increases by a factor of $\sqrt{m}$.

\section{Missing Proofs from \Cref{sec:reduction}}

\subsection{Proof of \Cref{thm:red-median-to-rank-agg}}

We use \Cref{alg:footrule-rank-agg}. The privacy guarantee follows from \Cref{cor:par-apx-median}, since the output is simply a post-processing of the estimates $(\tgamma_{j, q})_{j \in [m], q \in [m]}$. Thus, we will henceforth focus on proving the accuracy guarantees.

To prove the accuracy guarantees, let $\psi^* = \argmin_{\psi'} F(\psi', \Pi)$. It suffices to show the following claim: If $\|\gamma - \tgamma\|_\infty \leq \beta$, then $F(\psi, \Pi) - F(\psi^*, \Pi) \leq 2m\beta$. The explicit bounds for the three cases, then follows immediately from \Cref{cor:par-apx-median}.

To see this, note that
\begin{align*}
F(\psi, \Pi) &= \sum_{j \in [m]} \left(\frac{1}{n} \sum_{i \in [n]} |\psi(j) - \pi_i(j)|\right) \\
&= \sum_{j \in [m]} \gamma_{\psi(j), j}(\Pi) \\
&\leq m \beta + \sum_{j \in [m]} \tgamma_{\psi(j), j} \\
&\leq m \beta + \sum_{j \in [m]} \tgamma_{\psi^*(j), j} \\
&\leq 2 m \beta + \sum_{j \in [m]} \gamma_{\psi^*(j), j}(\Pi) \\
&= 2m\beta + F(\psi^*, \Pi)
\end{align*}
where the second inequality follows from the fact that $\psi$ is a minimum-weight matching of $G$.

\subsection{Proof of \Cref{thm:two-apx}}

This is an immediate consequence of \Cref{lem:reduction-generic} and \Cref{thm:red-median-to-rank-agg}.

\section{Missing Proofs from \Cref{sec:PTAS}}\label{sec:ptas-apx}
In this section, we prove \Cref{thm:ptas-merged}. It will be convenient for us to separate the pure-DP and approx-DP case into two theorems, as stated below.

\begin{theorem} \label{thm:ptas-main}
For every $\rho > 0, n, m \in \N$, and constant $\xi > 0$, there exists an efficient $\rho$-zCDP $\left(1 + \xi, O\left(\frac{m^{65/22} \sqrt{\log(m/\rho)}}{n \sqrt{\rho}}\right)\right)$-approximation algorithm for Kemeny rank aggregation.
\end{theorem}

\begin{theorem} \label{thm:ptas-main-pure}
For every $\varepsilon>0, n, m \in \N$, and constant $\xi>0$, there exists an efficient $\varepsilon$-DP $\left(1+\xi,O\left(\left(\frac{m^{27/7}\sqrt{\log (m/\varepsilon)}}{n\varepsilon}\right)\right)\right)$-approximation algorithm for Kemeny rank aggregation.
\end{theorem}

Note that, while the \Cref{thm:ptas-main} is stated in terms of zCDP instead of approx-DP, we can simply derive the approx-DP version by plugging in $\rho = O\left(\frac{\log(1/\delta)}{\eps}\right)$ and applying \Cref{thm:cdp-to-apx-dp}.

\subsection{The Small $n$ Case}

When $n$ is ``small'', our approximation algorithm guarantees are as follows.

\begin{theorem}
\label{thm:ptas-small-n}
For every $\rho > 0, n, m \in \N$ such that $n \leq \frac{m^{3/2}}{\rho^{1/2}}$, and any constant $\xi > 0$, there exists an efficient $\rho$-zCDP $\left(1 + \xi, O\left(\frac{m^{63/22}}{n^{10/11} \cdot \rho^{5/11}}\right)\right)$-approximation algorithm Kemeny rank aggregation.
\end{theorem}

\begin{theorem}
\label{thm:ptas-small-n-pure}
For every $\eps > 0, n, m \in \N$ such that $n \leq \frac{m^3}{\eps}$, and any constant $\xi > 0$, there exists an efficient $\varepsilon$-DP $\left(1 + \xi, O\left(\frac{m^{25/7}}{n^{6/7} \cdot \varepsilon^{6/7}}\right)\right)$-approximation algorithm for Kemeny rank aggregation.
\end{theorem}

Recall that $\tbw^\Pi=\tbs+\tbt$ where $\tbs, \tbt$ are as defined in \Cref{alg:ptas_small}.

To prove \Cref{thm:ptas-small-n-pure,thm:ptas-small-n}, by \Cref{lem:err-to-apx} and since clipping does not increase the $\ell_1$ distance by more than a factor of two (\Cref{obs:clipping}), it suffices to prove the following:

\begin{lemma} \label{lem:main-est-pairwise}
For every $\rho > 0, n, m \in \N$ such that $n \leq \frac{m^{3/2}}{\rho^{1/2}}$, there exists an efficient $\rho$-zCDP algorithm that produces an estimate $\tbw$ of $\bw^{\Pi}$ with $$\E[\|\tbw - \bw^{\Pi}\|_1] \leq O\left(\frac{m^{63/22}}{n^{10/11} \cdot \rho^{5/11}}\right).$$
\end{lemma}

\begin{lemma} \label{lem:main-est-pairwise-pure}
For every $\varepsilon > 0, n, m \in \N$ such that $n \leq \frac{m^{3}}{\varepsilon}$, there exists an efficient $\varepsilon$-DP algorithm that produces an estimate $\tbw$ of $\bw^{\Pi}$ with $$\E[\|\tbw - \bw^{\Pi}\|_1] \leq O\left(\frac{m^{25/7}}{n^{6/7} \cdot \varepsilon^{6/7}}\right).$$
\end{lemma}

\subsubsection{zCDP.} We will now prove \Cref{lem:main-est-pairwise} by instantiating \Cref{alg:ptas_small} with Gaussian noise and the two-way marginal algorithm from \Cref{thm:2-way-marginal}.

\begin{proof}[Proof of \Cref{lem:main-est-pairwise}]
Let $\sigma=\frac{m}{n\sqrt{\rho B}}$. Let $\cP$ denote the distribution on $\binom{[d]}{2}$ drawn as follows:
\begin{itemize}
\item Draw $u, v$ independently uniformly at random from $[m]$.
\item Draw $b_u < b_v$ uniformly from $[B]$ among $\binom{B}{2}$ such pairs.
\item Return $S = \{(u, b_u), (v, b_v)\}$ as a sample.
\end{itemize}

Run \Cref{alg:ptas_small} with $B = \left\lceil \frac{m^{3/11}}{n^{2/11} \rho^{1/11}} \right\rceil, \cD=\cN(0,\sigma^2)$ and \maralg~being the $\frac{\rho}{2}$-zCDP algorithm $\cA$ from \Cref{thm:2-way-marginal} with $\cP$ as specified. Since the $\ell_2$-sensitivities of $\bs$ is at most $\sqrt{m \cdot \frac{m}{B}} = \frac{m}{\sqrt{B}}$, \Cref{thm:gaussian} implies that the Gaussian mechanism applied to $\bs$ satisfies $\frac{\rho}{2}$-zCDP. Finally, the composition theorem (\Cref{thm:comp-CDP}) implies that the entire algorithm is $\rho$-zCDP.

As for the error guarantee, first notice that
\begin{align*}
t_{uv} &= \frac{1}{n} \sum_{i \in [n]} \bone[\iota(\pi_i(u)) < \iota(\pi_i(v))] \\
&= \frac{1}{n} \sum_{i \in [n]} \sum_{b_u, b_v \in [B] \atop b_u < b_v} \bone[\iota(\pi_i(u)) = b_u, \iota(\pi_i(v)) = b_v] \\
&= \sum_{b_u, b_v \in [B] \atop b_u < b_v} q_{\{(u, b_u), (v, b_v)\}}(\bX).
\end{align*}
As a result, we have
\begin{align*}
&\E[\|\tbt - \bt\|_1] \\
&\leq \E_{\cA}\left[\sum_{u, v \in [m]} \sum_{b_u, b_v \in [B] \atop b_u < b_v} \left|\tq_{\{(u, b_u), (v, b_v)\}} - q_{\{(u, b_u), (v, b_v)\}}(\bX)\right|\right] \\
&= \E_{\cA}\left[m^2 \binom{B}{2} \cdot \E_{S \sim \cP}\left[\left|\tq_S - q_S(\bX)\right| \right]\right] \\
&\lesssim \frac{m^{9/4} B^{9/4}}{\sqrt{n} \cdot \rho^{1/4}},
\end{align*}
where the last inequality is due to the guarantee of $\cA$ from \Cref{thm:2-way-marginal}.

For $\tbs$, we simply have
\begin{align*}
\E[\|\tbs - \bs\|_1] &= \sum_{u,v \in [m]} \E[|r_{uv}|] \lesssim m^2 \cdot \sigma \lesssim \frac{m^3}{n \sqrt{\rho B}}.
\end{align*}
Finally, note that
\begin{align*}
\E[\|\tbw - \bw^{\Pi}\|_1]
&\leq \E[\|\tbt - \bt\|_1] + \E[\|\tbs - \bs\|_1] \\
&\lesssim \frac{m^{9/4} B^{9/4}}{\sqrt{n} \cdot \rho^{1/4}} + \frac{m^3}{n \sqrt{\rho B}} \\
&\lesssim \frac{m^{63/22}}{n^{10/11} \cdot \rho^{5/11}},
\end{align*}
where the last inequality is due to our choice $B = \Theta\left(\frac{m^{3/11}}{n^{2/11} \rho^{1/11}}\right)$.
\end{proof}

\subsubsection{Pure-DP.}
The proof of \Cref{lem:main-est-pairwise-pure} is similar except that we use Laplace noise and the pure-DP two-way marginal algorithm from \Cref{cor:pure-2way}.

\begin{proof}[Proof of \Cref{lem:main-est-pairwise-pure}]
Let $b=\frac{m^2}{2 n \eps B}$.
Run \Cref{alg:ptas_small} with $B = \left\lceil \frac{m^{3/7}}{n^{1/7} \varepsilon^{1/7}} \right\rceil, \cD=\Lap(b)$ and \maralg~being the $\frac{\eps}{2}$-DP algorithm $\cA$ from \Cref{cor:pure-2way}. Since the $\ell_1$-sensitivities of $\bs$ is at most $m \cdot \frac{m}{B} = \frac{m^2}{B}$, \Cref{thm:laplace} implies that the Laplace mechanism applied to $\bs$ satisfies $\frac{\eps}{2}$-DP. Finally, the composition theorem (\Cref{thm:comp-pureDP}) implies that the entire algorithm is $\eps$-DP.

As for the error guarantee, recall from the proof of \Cref{lem:main-est-pairwise} that 
\begin{align*}
    t_{uv} = \sum_{b_u, b_v \in [B] \atop b_u < b_v} q_{\{(u, b_u), (v, b_v)\}}(\bX).
\end{align*}
As a result, we get 
\begin{align*}
    & \E[\|\tbt-\bt\|_1] \\
    &= \E\left[\sum_{u,v} \left| \sum_{b_u,b_v\in [B] \atop b_u<b_v} \tilde{q}_{\{(u,b_u),(v,b_v)\}}(\bX)-q_{\{(u,b_u),(v,b_v)\}}(\bX)\right|\right] \\
    & \le \E\left[\sum_{u,v} \sum_{b_u,b_v\in [B] \atop b_u<b_v} \left|\tilde{q}_{\{(u,b_u),(v,b_v)\}}(\bX)-q_{\{(u,b_u),(v,b_v)\}}(\bX)\right|\right] \\
    & \le \E\left[\sum_{S \in \binom{[d]}{2}} |\tq_S - q_S(\bX)|\right],
\end{align*}
which is at most $O\left(\sqrt{\frac{(mB)^5}{n\varepsilon}}\right) = O\left(\frac{m^{2.5}B^{2.5}}{\sqrt{n\varepsilon}}\right)$ by \Cref{cor:pure-2way}.

Moreover, we have
\begin{align*}
\E[\|\tbs - \bs\|_1] &= \sum_{u,v \in [m]} \E[|r_{uv}|] \lesssim m^2 \cdot b \lesssim \frac{m^4}{n \eps B}.
\end{align*}

Combining these, we get
\begin{align*}
\E[\|\tbw - \bw^{\Pi}\|_1]
&\leq \E[\|\tbt - \bt\|_1] + \E[\|\tbs - \bs\|_1] \\
&\lesssim \frac{m^{2.5}B^{2.5}}{\sqrt{n\eps}} + \frac{m^4}{n \eps B} \\
&\lesssim \frac{m^{25/27}}{n^{6/7} \cdot \eps^{6/7}},
\end{align*}
where the last inequality is due to our choice $B = \Theta\left(\frac{m^{3/7}}{n^{1/7} \varepsilon^{1/7}}\right)$.
\end{proof}
\subsection{The Large $n$ Case}

We will next handle the case of large $n$. The main theorem of this case is stated below.

\begin{theorem} \label{thm:ptas-large-n-cdp-expectation}
For any $\rho \in (0, 1], n, m \in \N$ such that $n \geq \frac{10000m\sqrt{\log(m/\rho)}}{\sqrt{\rho}}$, and constant $\xi > 0$, there exists an efficient $\rho$-zCDP $\left(1 + \xi, O\left(\frac{m^{2.5} \sqrt{\log m}}{n\sqrt{\rho}}\right)\right)$-approximation algorithm for Kemeny rank aggregation.
\end{theorem}
\begin{theorem} \label{thm:ptas-large-n-pure-expectation}
For any $\varepsilon \in (0, 1], n, m \in \N$ such that $n \geq \frac{400m^2{\log(m/\varepsilon)}}{{\varepsilon}}$, and constant $\xi > 0$, there exists an efficient $\varepsilon$-DP $\left(1 + \xi, O\left(\frac{m^{3.5} \sqrt{\log m}}{n{\varepsilon}}\right)\right)$-approximation algorithm for Kemeny rank aggregation.
\end{theorem}

In fact, it will be more convenient to prove the high-probability variants of the above theorems, as stated below.

\begin{theorem} \label{thm:ptas-large-n-cdp}
For any $\rho \in (0, 1], n, m \in \N$ such that $n \geq \frac{10000m\sqrt{\log(m/\rho)}}{\sqrt{\rho}}$, and constant $\xi > 0$, there exists an efficient $\rho$-zCDP algorithm that, with probability $1 - O\left(\frac{1}{m}\right)$ outputs an $\left(1 + \xi, O\left(\frac{m^{2.5} \sqrt{\log m}}{n\sqrt{\rho}}\right)\right)$-approximate solution for Kemeny rank aggregation.
\end{theorem}

\begin{theorem} \label{thm:ptas-large-n-pure}
For any $\varepsilon \in (0, 1], n, m \in \N$ such that $n \geq \frac{400m^2{\log(m/\varepsilon)}}{{\varepsilon}}$, and constant $\xi > 0$, there exists an efficient $\varepsilon$-DP algorithm that, with probability $1 - O\left(\frac{1}{m}\right)$ outputs an $\left(1 + \xi, O\left(\frac{m^{3.5} \sqrt{\log m}}{n{\varepsilon}}\right)\right)$-approximate solution for Kemeny rank aggregation.
\end{theorem}

Note that \Cref{thm:ptas-large-n-cdp-expectation} is now a consequence of applying \Cref{lem:high-prob-to-exp} with $\cM_1$ being the algorithm from \Cref{thm:ptas-large-n-cdp} and $\cM_2$ being an efficient $\frac{\rho}{3}$-zCDP\footnote{While \cite{AlabiGKM22} only states their guarantee in terms of approximate-DP, their algorithm only uses the Gaussian mechanism and thus satisfies zCDP as well.} $\left(1 + \xi, O\left(\frac{m^3}{n\sqrt{\rho}}\right)\right)$-approximation algorithm from \cite{AlabiGKM22}. Similarly, \Cref{thm:ptas-large-n-pure-expectation} is now a consequence of applying \Cref{lem:high-prob-to-exp} with $\cM_1$ being the algorithm from \Cref{thm:ptas-large-n-pure} and $\cM_2$ being an efficient $\frac{\eps}{3}$-DP $\left(1 + \xi, O\left(\frac{m^4}{\eps n}\right)\right)$-approximation algorithm from \cite{AlabiGKM22}.

The remainder of this section is devoted to proving \Cref{thm:ptas-large-n-cdp} and \Cref{thm:ptas-large-n-pure}.

\subsubsection{Common Lemma.}
The proof of both theorems will use \Cref{alg:ptas_large} but with different setting of the noise distribution $\cD$. We will thus start by stating a common lemma that will be used in both cases.

To do this, let us first define the  matrix $\tbw \in [0,1]^{m \times m}$ by
\begin{align*}
\tw_{uv} = 
\begin{cases}
w_{uv} - w_{vu} &\text{ if } (u, v) \in \Zl, \\
0 &\text{ if } (v, u) \in \Zl, \\
w_{uv} &\text{ otherwise.}
\end{cases}
\end{align*}
Note that this is simply the unnoised version of $\tbw'$ as defined in \Cref{alg:ptas_large}.

Let $\gamma$ be a positive real number, which will be specified later. Consider the following ``good'' events:
\begin{itemize}
\item $\cE_1$: $\tbw'$ is a bounded instance.
\item $\cE_2$: $\max_{\pi \in \bbS_m} |\cost_{\tbw'}(\pi) - \cost_{\tbw}(\pi)| < \gamma$.
\end{itemize}

Next we give a conditional guarantee of \Cref{alg:ptas_large}.
\begin{lemma}\label{lem:ptas_key}
If $\cE_1, \cE_2$ both hold, then \Cref{alg:ptas_large} outputs an $(1 + \xi, O(\gamma))$-approximate solution.
\end{lemma}

\begin{proof}
To do this, first notice that for any $\{u, v\} \in \binom{[n]}{2}$, $(u, v)$ or $(v, u)$ appears in $\Zl$ or $\Zs$ exactly once. This means that, for any permutation $\pi$, we have
\begin{align}
&\cost_{\bw}(\pi) \nonumber \\ &= \sum_{u,v\in [n]} \bone[\pi(u) < \pi(v)] \cdot w_{vu} \nonumber \\
&= \sum_{\{u,v\} \in \binom{[n]}{2}} \left(\bone[\pi(u) < \pi(v)] \cdot w_{vu} + \bone[\pi(v) < \pi(u)] \cdot w_{uv}\right) \nonumber \\
&= \sum_{(u, v) \in \Zl} \left(\bone[\pi(u) < \pi(v)] \cdot w_{vu} + \bone[\pi(v) < \pi(u)] \cdot w_{uv}\right) \nonumber \\
&\quad + \sum_{(u, v) \in \Zs} \left(\bone[\pi(u) < \pi(v)] \cdot w_{vu} + \bone[\pi(v) < \pi(u)] \cdot w_{uv}\right) \nonumber \\
&= \sum_{(u, v) \in \Zl} w_{vu} + \sum_{(u, v) \in \Zl} \left(\bone[\pi(v) < \pi(u)] \cdot (w_{uv} - w_{vu})\right) \nonumber \\
&\quad + \sum_{(u, v) \in \Zs} \left(\bone[\pi(u) < \pi(v)] \cdot w_{vu} + \bone[\pi(v) < \pi(u)] \cdot w_{uv}\right) \nonumber \\
&= \left(\sum_{(u, v) \in \Zl} w_{vu}\right) + \cost_{\tbw}(\pi), \label{eq:cost-expand-diff}
\end{align}
where the last equality is due to our definition of $\tbw$.

Let $\pi_{\out}$ denote the output of \Cref{alg:ptas_large}.
From $\cE_1$ and \Cref{thm:MS}, we have
\begin{align} \label{eq:approx-intermediate}
\cost_{\tbw'}(\pi_{\out}) \leq (1 + \xi) \cdot \cost_{\tbw'}(\pi^*).
\end{align}

Thus, we can bound $\cost_{\bw}(\pi_{\out})$ as follows:
\begin{align}
&\cost_{\bw}(\pi_{\out}) \nonumber\\ &\overset{\eqref{eq:cost-expand-diff}}{=}  \left(\sum_{(u, v) \in \Zl} w_{vu}\right) + \cost_{\tbw}(\pi_{\out}) \nonumber \\
&\overset{(\cE_2)}{\leq} \left(\sum_{(u, v) \in \Zl} w_{vu}\right) + \cost_{\tbw'}(\pi_{\out}) + \gamma \nonumber \\
&\overset{\eqref{eq:approx-intermediate}}{\leq} \left(\sum_{(u, v) \in \Zl} w_{vu}\right) + (1 + \xi) \cdot \cost_{\tbw'}(\pi^*) + \gamma \nonumber \\
&\overset{(\cE_2)}{\leq} \left(\sum_{(u, v) \in \Zl} w_{vu}\right) + (1 + \xi) \cdot \cost_{\tbw}(\pi^*) + (2 + \xi)\gamma \nonumber \\
&\overset{\eqref{eq:cost-expand-diff}}{\leq} (1 + \xi)\cdot\cost_{\bw}(\pi^*) + (2 + \xi) \gamma. \nonumber \qedhere
\end{align}
\end{proof}

\subsubsection{zCDP.}
We can now prove \Cref{thm:ptas-large-n-cdp} by using \Cref{alg:ptas_large} with $\cD$ being a Gaussian distribution, and bounding the probabilities that $\cE_1, \cE_2$ do not occur.

\begin{proof}[Proof of \Cref{thm:ptas-large-n-cdp}] 
Let $\sigma=\frac{m}{n\sqrt{\rho}}$. Run \Cref{alg:ptas_large} with $\cD=\cN(0,\sigma^2)$. Since the $\ell_2$-sensitivities of $\bw$ and $\tbw$ are at most $\frac{m}{n}$, \Cref{thm:gaussian} implies that the algorithm satisfies $\rho$-zCDP.

Due to \Cref{lem:ptas_key}, it suffices to prove that $\Pr[\neg \cE_1],\Pr[\neg \cE_2]\le O\left(\frac{1}{m}\right)$ for $\gamma=O\left(\frac{m^{2.5} \sqrt{\log m}}{n\sqrt{\rho}}\right)$.

\paragraph{Bounding $\Pr[\neg \cE_1]$.} From our assumption of $n$, we have $0.01 > \sigma\sqrt{100\log m}$. Thus, by standard concentration of Gaussian, with probability at least $1 - \frac{1}{m}$, we have $|\theta_{uv}|, |r_{uv}| \leq 0.01$ for all $u,v \in [m]$. When this holds, it is simple to verify that the instance $\tbw'$ is bounded. Thus, $\Pr[\neg \cE_1] \leq \frac{1}{m}$.

\paragraph{Bounding $\Pr[\neg \cE_2]$.} Consider any fixed $\pi$, we can write $\cost_{\tbw'}(\pi) - \cost_{\tbw}(\pi)$ as
\begin{align*}
&\cost_{\tbw'}(\pi) - \cost_{\tbw}(\pi) \\
&= \sum_{(u, v) \in \Zl} \bone[\pi(v) < \pi(u)] r_{uv} \\&\qquad + \sum_{(u, v) \in \Zs} (-1)^{\bone[\pi(u) < \pi(v)]} r_{uv}.
\end{align*}
Since $r_{uv}$'s are i.i.d. drawn from $\cN(0, \sigma^2)$, this last term is distributed as $\cN(0, (\sigma')^2)$ for some $\sigma' \leq m \cdot \sigma$. As a result, by standard concentration of Gaussian, we have
\begin{align*}
&\Pr[|\cost_{\tbw'}(\pi) - \cost_{\tbw}(\pi)| > 10m\sigma\sqrt{m\log m}] 
\\ &< 2\exp\left(-\left(10\sqrt{m\log m}\right)^2\right) \leq \frac{1}{m^{3m}}.
\end{align*}
Taking union bound over all $m!$ rankings $\pi \in S_m$, we have that $\Pr[\neg \cE_2] \leq \frac{1}{m}$ as well.
Finally, note that $\gamma$ here is $10m\sigma\sqrt{m\log (mn)} = O\left(\frac{m^{2.5} \sqrt{\log m}}{n\sqrt{\rho}}\right)$.

Therefore, by \Cref{lem:ptas_key}, the algorithm outputs an $\left(1 + \xi, O\left(\frac{m^{2.5} \sqrt{\log m}}{n\sqrt{\rho}}\right)\right)$-approximate solution with probability $1 - O\left(\frac{1}{m}\right)$.
\end{proof}

\subsubsection{Pure-DP.}
We use the same strategy as above except that now $\cD$ is a Laplace distribution. Note, however, that the sum of Laplace random variables is not distributed as a Laplace distribution, and thus its concentration is not clear a priori. Fortunately, \cite{ChanSS11} has shown concentration of such a sum, as stated more formally below.

\begin{theorem}[Lemma 2.8 in \cite{ChanSS11}] \label{thm:sum-laplace-concen}
Suppose $\kappa_1, \dots, \kappa_M$ are independent random variables, where each $\kappa_i$ has Laplace distribution $\mathrm{Lap}(b_i)$. Suppose $Y:=\sum_{i \in [M]}\kappa_i$ and $b_M:=\max_{i \in [M]}{b_i}$.
Let $\nu\ge \sqrt{\sum_{i \in [M]} b_i^2}$ and $0< \lambda< \frac{2\sqrt{2}\nu^2}{b_M}$. Then, $\Pr[Y>\lambda]\le \exp\left(-\frac{\lambda^2}{8\nu^2}\right)$.
\end{theorem}


\begin{proof}[Proof of \Cref{thm:ptas-large-n-pure}] 
Let $b=\frac{2m^2}{n{\varepsilon}}$. Run \Cref{alg:ptas_large} with $\cD=\mathrm{Lap}(b)$. Since the $\ell_1$-sensitivities of $\bw$ and $\tbw$ are at most $\frac{m^2}{n}$, \Cref{thm:laplace} implies that the algorithm satisfies $\eps$-DP.

Due to \Cref{lem:ptas_key}, it suffices to prove that $\Pr[\neg \cE_1],\Pr[\neg \cE_2]\le O\left(\frac{1}{m}\right)$ for $\gamma=O\left(\frac{m^{3.5} \sqrt{\log m}}{n{\varepsilon}}\right)$.

\paragraph{Bounding $\Pr[\neg\mathcal{E}_1]$.} Consider the concentration of Laplace distribution; the probability that each of $\theta_{uv}, r_{uv}$ it is not in range $[-0.01,0.01]$ is $\exp(-0.01/b)=\exp(-\frac{n}{100m^2})$.
If $n\ge 400m^2\log(m/\varepsilon)/\varepsilon$, then this is at most $\frac{1}{m^3}$. Taking a union bound over all these random variables, we have $\Pr[\neg\mathcal{E}_1]\le O\left(\frac{1}{m}\right)$.

\paragraph{Bounding $\Pr[\neg\mathcal{E}_2]$.} Similar to the previous proof, for any fixed $\pi$, we have
\begin{align*}
&\cost_{\tbw'}(\pi) - \cost_{\tbw}(\pi) \\
&= \sum_{(u, v) \in \Zl} \bone[\pi(v) < \pi(u)] r_{uv} \\ &\qquad + \sum_{(u, v) \in \Zs} (-1)^{\bone[\pi(u) < \pi(v)]} r_{uv},
\end{align*}
which is a sum of i.i.d. random variables drawn from $\Lap(b)$. We can use \Cref{thm:sum-laplace-concen} by setting $\nu=mb$ and $\lambda = 12bm\sqrt{m \log m}$. (Note here that $\nu \geq b\sqrt{m}$ and that $\nu^2/b = bm^2 \geq \lambda$ for any sufficiently large $m$.) Hence,
\begin{align*}
    &\Pr[|\cost_{\tbw'}(\pi) - \cost_{\tbw}(\pi)| > 12bm\sqrt{m\log m}] \\
    & \le 2\exp(-18m\log m) \le 2 m^{-3m}
\end{align*}
Taking union bound over all $m!$ rankings $\pi \in S_m$, we have that $\Pr[\neg \cE_2] \leq \frac{1}{m}$ as well. Again, here we let $\gamma = 12bm\sqrt{m\log m} = O\left(\frac{m^{3.5} \sqrt{\log m}}{n{\varepsilon}}\right)$.

Therefore, by \Cref{lem:ptas_key}, the algorithm outputs an $\left(1 + \xi, O\left(\frac{m^{3.5} \sqrt{\log m}}{n{\varepsilon}}\right)\right)$-approximate solution with probability $1 - O\left(\frac{1}{m}\right)$.
\end{proof}

\newpage

\subsection{Putting Things Together: Proofs of \Cref{thm:ptas-main,thm:ptas-main-pure}}


We can now prove the main theorems (\Cref{thm:ptas-main,thm:ptas-main-pure}) by simply combining the small and large cases.

\begin{proof}[Proof of \Cref{thm:ptas-main}]
This simply results from combining \Cref{thm:ptas-large-n-cdp-expectation} and \Cref{thm:ptas-small-n}. Namely, if $n \geq \Omega\left(\frac{m\sqrt{\log(m/\rho)}}{\sqrt{\rho}}\right)$, then we can simply run the algorithm from \Cref{thm:ptas-large-n-cdp-expectation} which immediately yields the desired approximation guarantee. Otherwise, if $n \leq O\left(\frac{m\sqrt{\log(m/\rho)}}{\sqrt{\rho}}\right)$, we can run the algorithm from \Cref{thm:ptas-small-n} instead. This gives an additive error of
\begin{align*}
\frac{m^{63/22}}{n^{10/11} \cdot \rho^{5/11}} &\lesssim \frac{m^{63/22}}{n \cdot \rho^{5/11}} \cdot \left(\frac{m\sqrt{\log(m/\rho)}}{\sqrt{\rho}}\right)^{1/11} \\
&\lesssim \frac{m^{65/22} \sqrt{\log(mn/\rho)}}{n \sqrt{\rho}},
\end{align*}
as desired.
\end{proof}

\begin{proof}[Proof of \Cref{thm:ptas-main-pure}]
This simply results from combining \Cref{thm:ptas-large-n-pure-expectation} and \Cref{thm:ptas-small-n-pure}. Namely, if $n \geq \Omega\left(\frac{m^2\sqrt{\log(m/\varepsilon)}}{{\varepsilon}}\right)$, then we can simply run the algorithm from \Cref{thm:ptas-large-n-pure-expectation} which immediately yields the desired approximation guarantee. Otherwise, if $n \leq O\left(\frac{m^2\sqrt{\log(m/\varepsilon)}}{{\varepsilon}}\right)$, we can run the algorithm from \Cref{thm:ptas-small-n-pure} instead. This gives an additive error of
\begin{align*}
    \frac{m^{25/7}}{n^{6/7}\varepsilon^{6/7}} & \lesssim \frac{m^{25/7}}{n\varepsilon^{6/7}}\cdot \left(\frac{m^2\sqrt{\log (m/\varepsilon)}}{\varepsilon}\right)^{1/7} \\
    & \lesssim \frac{m^{27/7}\sqrt{\log(mn/\varepsilon)}}{n\varepsilon}.
\end{align*}
as desired.
\end{proof}

\fi

\end{document}